\documentclass{new_tlp}
\usepackage{mathptmx}

\usepackage{latexsym}
\usepackage{amssymb}
\usepackage{amsmath}
\usepackage{lscape}
\usepackage{url}
\usepackage{adjustbox}
\newtheorem{definition}{Definition} % [section]
\newtheorem{example}{Example} % [section]
\newtheorem{lemma}{Lemma} % [section]
\newtheorem{proposition}{Proposition} % [section]
\newtheorem{theorem}{Theorem} % [section]
\newtheorem{corollary}{Corollary} % [section]

\usepackage{color}

\newcommand{\red}[1]{{\color{red} #1}}

\newcommand{\blue}[1]{{\color{blue} #1}}



\newcommand*{\ie}{\textit{i.e.}, }
\newcommand*{\eg}{\textit{e.g.}, }
\newcommand{\sep}{\mathit{\;]\![\;}}

\newcommand{\lra}{\longrightarrow}
\newcommand{\lrac}{\hookrightarrow}
\newcommand{\lracn}{\Longrightarrow}

\newcommand{\term}{$\cT\!\!{erm}$}
\newcommand{\var}{{\cV}ar}
\newcommand{\emptyseq}{\varepsilon} % empty sequence
\renewcommand{\emptyset}{\{\}} %%{\varnothing}
\renewcommand{\phi}{\varphi}
\newcommand{\ol}[1]{\overline{#1}}  % sequence of objects

\newcommand{\true}{\mathsf{true}}
\newcommand{\false}{\mathsf{false}}
\newcommand{\clauses}{\mathsf{rules}}
\newcommand{\rules}{\mathsf{rules}}

\def \tuplec#1{\langle\!| #1 |\!\rangle}
\newcommand{\mystate}[2]{\langle{#1}\,|\,{#2}\rangle}
\newcommand{\catom}{\operatorname{\mathit{c-atom}}}
\newcommand{\set}{\operatorname{\mathit{Set}}}

\newcommand{\alts}{\mathit{alts}}

\newcommand{\negcon}{\mathit{neg\_constr}}

\newcommand{\TR}{{\mathtt{TR}}}
\newcommand{\TC}{{\mathtt{TC}}}
\newcommand{\PTC}{{\mathtt{PTC}}}

\long\def\comment#1{}

\def\defemb#1#2{\expandafter\def\csname #1\endcsname
                              {\relax\ifmmode #2\else\hbox{$#2$}\fi}}
\defemb{cA}{{\cal A}}
\defemb{cB}{{\cal B}}
\defemb{cC}{{\cal C}}
\defemb{cD}{{\cal D}}
\defemb{cE}{{\cal E}}
\defemb{cF}{{\cal F}}
\defemb{cG}{{\cal G}}
\defemb{cH}{{\cal H}}
\defemb{cI}{{\cal I}}
\defemb{cJ}{{\cal J}}
\defemb{cL}{{\cal L}}
\defemb{cM}{{\cal M}}
\defemb{cN}{{\cal N}}
\defemb{cO}{{\cal O}}
\defemb{cP}{{\cal P}}
\defemb{cQ}{{\cal Q}}
\defemb{cR}{{\cal R}}
\defemb{cS}{{\cal S}}
\defemb{cT}{{\cal T}}
\defemb{cU}{{\cal U}}
\defemb{cV}{{\cal V}}
\defemb{cX}{{\cal X}}
\defemb{cZ}{{\cal Z}}

\begin{document}

\title[Concolic Testing in CLP]{Concolic Testing in CLP}

\author[F.~Mesnard, \'E.~Payet and G.~Vidal]
{FRED MESNARD, \'ETIENNE PAYET \\
  LIM - Universit\'e de la R\'eunion, France\\
  \email{\{frederic.mesnard, etienne.payet\}@univ-reunion.fr}
\and
GERM\'AN VIDAL
\thanks{This author has been partially supported by EU (FEDER) and
    Spanish MCI/AEI under grants TIN2016-76843-C4-1-R and
    PID2019-104735RB-C41, and by the \emph{Generalitat Valenciana}
    under grant Prometeo/2019/098 (DeepTrust).}
\\
MiST, VRAIN, Universitat Polit\`ecnica de Val\`encia\\
\email{gvidal@dsic.upv.es}
}

% \pagerange{\pageref{firstpage}--\pageref{lastpage}}
% \volume{\textbf{10} (3):}
% \jdate{March 2002}
% \setcounter{page}{1}
% \pubyear{2002}

\maketitle

\label{firstpage}

\begin{abstract}
  Concolic testing is a popular software verification technique based
  on a combination of concrete and symbolic execution. Its main focus
  is finding bugs and generating test cases with the aim of maximizing
  code coverage.
  A previous approach to concolic testing in logic programming was
  not sound because it only dealt with positive constraints (by means
  of substitutions) but could not represent negative constraints.
  In this paper, we present a novel framework for concolic testing of
  CLP programs that generalizes the previous technique. In the CLP
  setting, one can represent both positive and negative constraints in
  a natural way, thus giving rise to a sound and (potentially) more
  efficient technique.
  Defining verification and testing techniques for CLP programs is
  increasingly relevant since this framework is becoming popular as an
  intermediate representation to analyze programs written in other
  programming paradigms.\\[2ex]
 \emph{This paper is under consideration for acceptance in Theory and Practice of Logic Programming (TPLP).}
\end{abstract}


%%%%%%%%%%%%%%%%%%%%%%%%%%%%%%%%%%%
% Intro
%%%%%%%%%%%%%%%%%%%%%%%%%%%%%%%%%%%
\section{Introduction}\label{sec:intro}

Symbolic execution was first proposed by King \citeyear{Kin76} as a technique
for automated test case generation. Essentially, the program is run
with some unknown (symbolic) input data. Symbolic execution then
proceeds by speculatively exploring all possible computations.
Let us consider a simple imperative language with \emph{conditionals}
and that the \emph{trace} of an execution is denoted by the sequence
of choices made in the conditionals of this execution (\eg the trace
$\mathtt{tft}$ denotes that execution entered the \texttt{true} branch
of the first conditional, then the \texttt{false} branch of the second
conditional, and finally the \texttt{true} branch of the third
conditional).

During symbolic execution, whenever a conditional with condition $c$
is found, one should explore both branches. In one of the branches,
$c$ is assumed; in the other branch, one can assume the negation of
this condition \ie $\neg c$. By gathering all the constraints
assumed in a symbolic execution, and solving them, one can produce
values for the input arguments.
% In this way, one can (automatically)
% produce \emph{test cases} that will follow the execution paths of the
% considered symbolic executions.
%
Symbolic execution methods are \emph{sound} in the following sense: if
a symbolic execution with trace $\pi$ collects constraints
$c_1,\ldots,c_n$, then solving $c_1\land \ldots\land c_n$ will produce
values for a concrete call whose execution will have the same trace
$\pi$ (\ie it will follow the same execution path of the symbolic
execution that produced these constraints).
This is a key property in order to achieve a good program
coverage. Note that test case generation based on symbolic execution
is in principle aimed at a full path coverage.

Concolic testing \cite{GKS05,SMA05} can be seen as an evolution of
test case generation methods based on symbolic execution. The main
difference is that, now, both concrete and symbolic executions are
performed in parallel (thus the term ``concolic'': \emph{conc}rete $+$
symb\emph{olic}).
Roughly speaking, concolic testing proceeds iteratively as follows.
It starts with an arbitrary concrete call. Then, this call is executed
with the standard semantics, together with a corresponding symbolic
call that mimics the execution of the concrete one. This is called a
concolic execution. Once this concolic execution terminates, one can
produce alternative test cases by negating some of the collected
constraints and, then, solving them. For example, if we gathered the
sequence of constraints $c_1,c_2,c_3$ (\eg associated to the execution
of three conditionals) with associated trace $\mathtt{ttt}$, we can
now solve the constraints $\neg c_1$ (trace $\mathtt{f}$),
$c_1\land \neg c_2$ (trace $\mathtt{tf}$) and
$c_1\land c_2\land \neg c_3$ (trace $\mathtt{ttf}$) in order to
produce three new, alternative test cases that will follow a different
execution path. A new iteration starts by considering any of the new
test cases, and so forth. In principle, the process terminates when
all alternative test cases have been processed. Nevertheless, the
search space is typically infinite (as in symbolic execution based
methods).

Concolic execution has gained popularity because of some advantages
over the symbolic execution based methods. For instance, one can
automatically detect some run-time errors since concolic testing
performs standard (concrete) executions and, thus, if some error is
spotted, we know that this is an actual run-time error. Furthermore, when the
constraints become too complex for state-of-the-art solvers and the
methods based on symbolic execution just give up, concolic testing
can still inject some concrete data (from the concrete component) and
simplify the constraints in order to make them tractable.

% In general, one expects both symbolic execution to be \emph{sound} in
% the following sense. Let us consider a simple imperative language with
% conditionals. Then, if we have a symbolic execution with trace
% $\pi$ that produces the sequence of constraints $c_1,c_2,c_3$, we
% expect that solving $c_1\land c_2\land c_3$ will produce input values
% so that a concrete execution with these values will have the same
% trace.

Although concolic testing is quite popular in imperative and
object-oriented languages, only a few works can be found in the context
of functional and logic programming languages. Some notable exceptions
are those of Giantsios et al.\ \citeyear{GPS15} and
Palacios and Vidal \citeyear{PV15} for a functional language,
and those of Vidal \citeyear{Vid14} and Mesnard et al.\ \citeyear{MPV15}
for a logic language. In the context of logic programming, concolic
execution becomes particularly challenging because computing the
alternatives of a predicate call is not as straightforward as in
imperative programming, where negating a condition suffices. Consider,
\eg a predicate call that matches rules $r_1$ and $r_2$. Here, a
full path coverage should include test cases for all the following
alternatives: no rule is matched; only rule $r_1$ is matched; only
rule $r_2$ is matched; and both rules $r_1$ and $r_2$ are matched
(assuming all these cases are feasible). The problem of finding all
these alternative test cases is based on so-called \emph{selective
  unification} \cite{MPV15,MPV17}.

% Loosely speaking, the problem with the algorithm in \cite{MPV15} is
% that it is not \emph{sound} in the following sense: whenever we gather
% some constraints $c_1,c_2,\ldots,c_n$ along some execution path, one
% would expect that solving these constraints will produce input values
% that will lead to a concrete initial call which follows the same
% execution path. Unfortunately, this is not true for the concolic
% testing algorithm in \cite{MPV15}.

A limitation of the approach to concolic testing of Mesnard et al.\
\citeyear{MPV15} is that only \emph{positive} constraints (represented as
substitutions) are gathered during concolic execution. As a
consequence, the algorithm is not sound in the above sense, as
witnessed by the following example:

\begin{example}
  Let us consider the following simple program:
  \[
    \begin{array}{lll}
      p(f(a)). & (r_1)\\
      p(f(X))\leftarrow q(X). & (r_2)\\
      q(b). & (r_3) \\
    \end{array}
  \]
  where terms are built, \eg from constants $a,b,c$ and the unary
  function symbol $f$.
  If we consider a semantics that only computes the first solution of
  a goal (as in the approach by Mesnard et al.\ \citeyear{MPV15}), the only \emph{feasible} execution
  paths for an initial goal that calls predicate $p$ are the
  following:
  \begin{itemize}
  \item A call that matches no rule, \eg $p(a)$.
  \item A call that matches both rules $r_1$ and $r_2$ and then
    succeeds, \eg $p(f(a))$.
  \item A call that matches only rule $r_2$ and, then, calls predicate
    $q$ and matches rule $r_3$, \eg $p(f(b))$.
  \item A call that matches only rule $r_2$ and, then, calls predicate
    $q$ but does not match rule $r_3$, \eg $p(f(c))$.
  \end{itemize}
  However, the concolic testing procedure of Mesnard et al.\ \citeyear{MPV15} may fail to
  compute the last test case.
  For instance, let us consider that the process starts with the
  initial call $p(a)$, which matches no rule. Now, the computed
  alternatives could be $p(f(a))$ that matches both $r_1$ and $r_2$
  and $p(f(b))$ that only matches $r_2$.\footnote{Note that matching
    only $r_1$ is not feasible in this case. \textit{E.g.}, there is
    no call of the form $p(t)$ for some term $t$ such that $p(t)$ matches
    rule $r_1$ but not $r_2$.}
  Let us first consider $p(f(a))$. This call immediately succeeds, so
  there are no more alternatives to be computed.
  Consider now the call $p(f(b))$. This call first matches rule $r_2$
  and, then, calls $q(b)$, which succeeds.
  Here, one can still generate a new alternative test case: one that
  (only) matches rule $r_2$ and, then, fails to match rule $r_3$.
  Unfortunately, the concolic testing algorithm of Mesnard et al.\ 
  \citeyear{MPV15} may generate $p(f(a))$ again since it only knows that
  the argument of $p$ must unify with $f(X)$ (to match rule $r_2$) and
  that $X$ must not unify with $b$ (to avoid matching rule $r_3$). Thus,
  $p(f(a))$ is a solution. However, this is not the solution we
  expected, since this call will match rule $r_1$ and succeed
  immediately.
\end{example}
%
% Intuitively, the reason for the unsoundness of \cite{MPV15} is that
% only positive constraints are gathered (in the form of substitutions),
% but negative constraints are just ignored. Therefore, in the example
% above, the algorithm in \cite{MPV15} does not \emph{remember} that $X$
% must be different from $a$ when entering rule $r_2$.
%
% does not remember that $X$ must
% be different from $a$ when entering rule $r_2$.
%
% We solve this problem
% in the context of CLP, where both positive and negative constraints
% can be represented in a natural way.
%
In this work, we consider the development of a concolic testing
framework for CLP programs, where both positive and \emph{negative}
constraints can be represented in a natural way.
Our main contributions are the following:
\begin{itemize}
\item We extend the original framework \cite{MPV15} to CLP
  programs. In particular, we illustrate our approach with two
  instances: CLP(\term) and CLP($\cN$). %CLP($\mathcal{Q}_{\mathit{lin}}$).
  As an advantage of this formulation, efficient external constraint solvers can
  be used to produce test cases.

\item In contrast to previous approaches, we prove the soundness of
  our approach, \ie whenever a test case for a given execution path
  is produced, we can ensure that the execution of this test case will
  indeed follow the associated path. This can be ensured thanks to the
  use of \emph{negative} constraints.

\item We prove that, if the constraint domain is decidable, then the
so-called selective unification problem is decidable too. Thus we
extend the results of Mesnard et al.\ \citeyear{MPV17}.
\end{itemize}
Defining verification and testing techniques for CLP programs is increasingly
relevant since this setting is becoming popular as an intermediate
representation to analyze programs written in other programming
paradigms, see, \eg the work of Gange et al.\ \citeyear{GNSSS15} 
and Gurfinkel et al.\ \citeyear{GKKN15}.
Furthermore, concolic testing may be useful in the context of
run-time verification techniques; see, \eg the work of
Stulova et al.\ \citeyear{SMH14}.
Therefore, our approach to
concolic testing may constitute a significant contribution to these
research areas.

Some more details and proofs of technical results can be found in the
Appendix.

%%%%%%%%%%%%%%%%%%%%%%%%%%%%%%%%%%%
% Preliminaries
%%%%%%%%%%%%%%%%%%%%%%%%%%%%%%%%%%%
\section{Preliminaries}\label{sec:prel}

We assume some familiarity with the standard definitions and notations
for logic programming as introduced by Apt \citeyear{Apt97} and for constraint
logic programming as introduced %in~\cite{JM94,JMMS98}.
by Jaffar et al.\ \citeyear{JMMS98}.
Nevertheless, in order to make the paper as self-contained as possible,
we present in this section the main concepts which are needed to understand
our development.

We denote by $|S|$ the cardinality of the set $S$ and by $\mathbb{N}$
the set of  natural numbers.
%
%%% SIGNATURE
From now on, we fix an infinite countable set $\cV$ of \emph{variables}
together with a \emph{signature} $\Sigma$, \ie a pair $\langle F,\Pi_C\rangle$
where $F$ is a finite set of \emph{function symbols} and $\Pi_C$ is
a finite set of \emph{predicate symbols} with $F\cap\Pi_C=\emptyset$
and $(F\cup\Pi_C)\cap\cV=\emptyset$.
Every element of $F\cup\Pi_C$ has an \emph{arity} which is the number
of its arguments. We write $f/n\in F$ (resp. $p/n\in \Pi_C$) to denote
that $f$ (resp. $p$) is an element of $F$ (resp. $\Pi_C$) whose arity
is $n\geq 0$. A \emph{constant symbol} is an element of $F$ whose arity
is 0.

%%% TERMS
A \emph{term} is a variable, a constant symbol or an entity
$f(t_1,\dots,t_n)$ where $f/n\in F$, $n\geq 1$ and $t_1,\dots,t_n$
are terms. For any term $t$, we let $\var(t)$ denote the set of variables
occurring in $t$. This notation is naturally extended to sets of terms.
We say that $t$ is \emph{ground} when $\var(t)=\emptyset$.
% Positions are used to address the nodes of a term viewed as a tree. A
% \emph{position} $p$ in $t$, in symbols $p\in\pos(t)$, is represented by a
% finite sequence of natural numbers, where $\emptyseq$ denotes the
% root position.
% %
% We let $t|_p$ denote the \emph{subterm} of $t$ at position $p$ and
% $t[s]_p$ the result of \emph{replacing the subterm} $t|_p$ by the term
% $s$.
% %
% The depth $\depth(t)$ of $t$ is defined as: $\depth(t) = 0$ if
% $t$ is a variable and $\depth(f(t_1,\ldots,t_n)) =
% 1+\mathsf{max}(\depth(t_1),\ldots,\depth(t_n))$, otherwise.
% %
% We say that $t|_p$ is a subterm of $t$ at depth $k$ if there are $k$
% nested function symbols from the root of $t$ to the root of $t|_p$.

%%% FIRST-ORDER FORMULAE
An \emph{atomic constraint} is an element $p/0$ of $\Pi_C$ or an entity
$p(t_1,\dots,t_n)$ where $p/n\in\Pi_C$, $n\geq 1$ and $t_1,\dots,t_n$
are terms. A first-order \emph{formula} on $\Sigma$ is built from atomic
constraints in the usual way using the logical connectives $\land$, $\lor$,
$\lnot$, $\rightarrow$, $\leftrightarrow$ and the quantifiers $\exists$ and
$\forall$. For any formula $\phi$, we let $\var(\phi)$ denote its set of
free variables and $\exists \phi$ (resp. $\forall \phi$) its
existential (resp. universal) closure.

%%% INTERPRETATIONS
We fix a \emph{$\Sigma$-structure} $\cD$, \ie a pair
$\langle D,[\cdot]\rangle$ which is an interpretation of the symbols
in $\Sigma$. The set $D$ is called the \emph{domain} of $\cD$ and $[\cdot]$
maps each $f/0\in F$ to an element of $D$, each $f/n\in F$ with $n\geq 1$
to a function $[f]: D^n \rightarrow D$, each $p/0\in \Pi_C$ to an element of
$\{0,1\}$, and each $p/n\in\Pi_C$ with $n\geq 1$ to a boolean function
$[p]: D^n \rightarrow \{0,1\}$.
We assume that the binary predicate symbol $=$ is in $\Sigma$ and
is interpreted as identity in $D$.
% Valuation:
A \emph{valuation} is a mapping from $\cV$ to $D$.
Each valuation $v$ extends by morphism to terms.
A valuation $v$ induces a valuation $[\cdot]_v$ of
terms to $D$ and of formulas to $\{0,1\}$.

% Models:
Given a formula $\phi$ and a valuation $v$, we write $\cD\models_v \phi$
when $[\phi]_v=1$. We write $\cD\models \phi$ when $\cD\models_v \phi$ for
all valuations $v$. Notice that $\cD\models \forall \phi$ if and only if
$\cD\models \phi$, that $\cD\models \exists \phi$ if and only if there exists
a valuation $v$ such that $\cD\models_v \phi$, and that
$\cD\models \lnot\exists \phi$ if and only if $\cD\models \lnot \phi$.
We say that a formula $\phi$ is \emph{satisfiable} (resp.
\emph{unsatisfiable}) in $\cD$ when $\cD\models \exists \phi$ (resp.
$\cD\models \lnot \phi$).

%%% CONSTRAINTS
We fix a set $\cL$ of admitted formulas, the elements of which are
called \emph{constraints}.
In this paper, we suppose that $\cL$  contains
all the atomic constraints, the always satisfiable constraint $\true$
and the unsatisfiable constraint $\false$, and any quantified boolean combination
of such formulae (while usually $\cL$ only contains
conjunctions of atomic constraints which are implicitly existentially quantified).
% The constraint solver solv:
We assume that there is a computable function $\mathit{solv}$ which maps each
$c\in\cL$ to one of \texttt{true} or \texttt{false} indicating whether $c$
is satisfiable or unsatisfiable in $\cD$.
In particular, it implies that the constraint domain has to be decidable.
We call $\mathit{solv}$ the \emph{constraint solver}.

% \comment{%%%%%%%%%%%%%%%%%%%
% Example: naturals
\begin{example}[CLP($\cN$) and CLP(\term)]\label{example-naturals}
  The constraint domain $\cN$ has $<$, $\leq$, $=$, $\neq$, $\geq$,
  $>$ as predicate symbols, $+$  as function symbol and
  sequences of digits as constant symbols.
  The domain of computation is the structure with the set of naturals,
  denoted by $\mathbb{N}$, as domain and where the predicate symbols and
  the function symbol are interpreted as the usual relations and function
  over the naturals. A constraint solver for $\cN$
  %always returning either {\tt true} or {\tt false}
  is described by, \eg Comon and Kirchner \citeyear{CK99}.

\comment{%%%%%%%%%%%%%%%%%%%
% Example: rationals
\begin{example}[CLP($\cQ$)]\label{example-rationals}
  The constraint domain $\cQ$ has $<$, $\leq$, $=$, $\neq$, $\geq$,
  $>$ as predicate symbols, $+$, $-$  as function symbols and
  sequences of digits as constant symbols.
  %$\cL$  is the set of conjunctions of linear atomic constraints.
  The domain of computation is the structure with the set of rationals,
  denoted by $\mathbb{Q}$, as domain and where the predicate symbols and
  the function symbols are interpreted as the usual relations and functions
  over the rationals. A constraint solver for $\cQ$
  is described in, \eg \cite{FR75}.
  %\qed
\end{example}
}%%%%%%%%%%%%%%%%%%%%%%%%%%%%

% \comment{%%%%%%%%%%%%%%%%%%%
  The constraint domain \term{} has $=$, $\neq$ as predicate symbols and strings
  of alphanumeric characters as function symbols. The domain of computation
  is the set of \emph{finite trees} (or, equivalently, of finite
  terms), $\mathit{Tree}$. The interpretation of a constant is a tree with a
  single node labeled with the constant. The interpretation of an $n$-ary
  function symbol $f$ is the function
  $f_{\mathit{Tree}}:\mathit{Tree}^n \rightarrow \mathit{Tree}$
  mapping the trees $T_1$, \dots, $T_n$ to a new tree with root labeled
  with $f$ and with $T_1$, \dots, $T_n$ as child nodes. A constraint solver
  for \term{} is also described in \cite{CK99}.
\end{example}
%}%%%%%%%%%%%%%%%%%%%%%%%%%%%

%%% SEQUENCES
We let $\ol{o_n}$ denote the finite sequence of syntactic objects
$o_1,\ldots,o_n$; we also write $\ol{o}$ when the number of
elements is not relevant. We let $\emptyseq$ denote the empty
sequence and $\ol{o}, \ol{o'}$ denote the concatenation of
sequences $\ol{o}$ and $\ol{o'}$.
Sequences of distinct variables are denoted by $\ol{X}$, $\ol{Y}$ or
$\ol{Z}$ and are sometimes considered as sets of variables.
%: we may write $\forall_{\ol{X}}$, $\exists_{\ol{X}}$ or $\ol{X}\cup\ol{Y}$.
Sequences
of (not necessarily distinct) terms are denoted by $\ol{s}$, $\ol{t}$ or
$\ol{u}$.
Given two sequences of $n$ terms $\ol{s_n}$ and $\ol{t_n}$, we write
$\ol{s_n}=\ol{t_n}$ to denote the constraint
$s_1=t_1 \land \dots \land s_n=t_n$.
We also extend the notation $[\cdot]_v$ by letting $[\ol{s_n}]_v$ denote
the sequence $[s_1]_v,\ldots,[s_n]_v$.

%$
\comment{%%%%%%%%%%%%%%%%%%%%%%%%%%%%%%%%%%%%%%%%%%%%
%%% QUANTIFIER AND VARIABLE ELIMINATION
A structure $\cD$ admits \emph{quantifier elimination}
if for each first-order formula $\phi$ there exists a quanti\-fier-free
formula $\psi$ such that $\cD \models \phi \leftrightarrow \psi$.
$\cD$ admits \emph{variable elimination} if for each quantifier-free
formula  $\phi(\ol{X},Y)$ there exists a quantifier-free formula
$\psi(\ol{X})$ such that
$\cD \models \exists Y \phi(\ol{X},Y) \leftrightarrow  \psi(\ol{X})$.
For instance, $\cQ$ admits variable elimination
via the Fourier-Motzkin algorithm, see, \eg \cite{MS98}.
}%%%%%%%%%%%%%%%%%%%%%%%%%%%%%%%%%%%%%%%%%%%%%%%%

%%% CLP
% Signature
The signature in which all programs and queries under consideration are
included is $\Sigma_L = \langle F, \Pi_C\cup\Pi_P\rangle$ where $\Pi_P$
is the set of predicate symbols that can be defined in programs, with
$\Pi_C\cap\Pi_P=\emptyset$.
% Atoms, states, rules, programs:
An \emph{atom} has the form $p(\ol{s_n})$ where $p/n\in\Pi_P$ and $\ol{s_n}$
is a sequence of terms. The definitions and notations on terms
($\var$, ground,\ldots) are extended to atoms in the natural way.
We write $[p(\ol{s_n})]_v$ to denote the atom $p([\ol{s_n}]_v)$.
For any sequence $\ol{A_m}$ of atoms we let
$[\ol{A_m}]_v$ denote the sequence $[A_1]_v,\ldots,[A_m]_v$.
A \emph{rule} has the form $H \leftarrow c \land \ol{B}$
where $H$ is an atom called the \emph{head} of the rule, $c$ is a satisfiable
constraint and $\ol{B}$ is a finite sequence of atoms. For the sake of
readability, in examples we may simplify rules of the form
$H \leftarrow c \land \emptyseq$ to $H \leftarrow c$.
A \emph{program} is a finite set of rules.
A \emph{state} has the form $\mystate{d}{\ol{B}}$ where $\ol{B}$ is a
finite sequence of atoms and $d$ is a constraint. A
\emph{constraint atom} is a state of the form $\mystate{d}{p(\ol{t})}$.
% A constraint atom of the form $\mystate{c}{p(\ol{X})}$ is \emph{projected}
% when $\var(c)\subseteq \{\ol{X}\}$.
We denote states as $Q$, $Q'$\ldots{} or
$R$, $R'$\ldots{} and constraint atoms as $C$, $C'$\ldots{}
For any state $Q:=\mystate{d}{\ol{B}}$ and any constraint $d'$, we let
$Q \land d'$ denote the state $\mystate{d \land d'}{\ol{B}}$.

Any state can be seen as a finite description of a possibly infinite set
of sequences of atoms, the arguments of which are values from $D$. More
precisely, the \emph{set described by} a state $Q:=\mystate{d}{\ol{B}}$
is defined as
$\set(Q) = \left\{ [\ol{B}]_v \mid \cD \models_v d \right\}$.
For instance, for $Q:=\mystate{Y\leq X+2}{p(X),q(Y)}$ in
$\cN$, we have $p(0),q(2)\in\set(Q)$.
For any states $Q:=\mystate{c}{\ol{A}}$ and $Q':=\mystate{d}{\ol{B}}$,
we say that $Q'$ is \emph{less instantiated} (or \emph{more general})
than $Q$ (equivalently, that $Q$ is more restricted than $Q'$), and we
write $Q\leq Q'$, when $\ol{A}$ and $\ol{B}$ are variants and,
moreover, $\set(Q)\subseteq \set(Q')$; furthermore, we say they are
\emph{equivalent} when $\set(Q) = \set(Q')$ (instead of $\set(Q)\subseteq \set(Q')$).
Furthermore, we say that $Q$ and $Q'$ are \emph{equivalent}, and we write
$Q \equiv Q'$, when $\set(Q) = \set(Q')$.

%%% CLP OPERATIONAL SEMANTICS
We consider the usual operational semantics given in terms of
\emph{derivations} from states to states. Let
$\mystate{d}{p(\ol{u}),\ol{B}}$ be a state and
$p(\ol{s})\leftarrow c \land \ol{B'}$ be a fresh copy of a rule $r$.
When $\mathit{solv}(\ol{s} = \ol{u} \land c \land d) = \mathtt{true}$
then (in this work, a fixed leftmost selection rule is assumed)
\[
  \mystate{d} {p(\ol{u}),\ol{B}}\lra_r
  \mystate{\ol{s}=\ol{u}\land c\land d}{\ol{B'},\ol{B}}
\]
is a  \emph{derivation step} of $\mystate{d}{p(\ol{u}),\ol{B}}$
with respect to $r$ with $p(\ol{s})\leftarrow c \land \ol{B'}$ as its
\emph{input rule}. A state $Q:=\mystate{d}{\ol{B}}$ is said to be
\emph{successful} if $\ol{B}$ is empty; it is said to be
\emph{failed} if $\ol{B}$ is not empty and no derivation step is
possible.
We write $Q \lra_P^+ Q'$ to summarize a finite number ($> 0$) of derivation
steps from $Q$ to $Q'$ where each input rule comes from program $P$.
Let $Q_0$ be a state. A sequence of derivation steps
$Q_0 \lra_{r_1} Q_1 \lra_{r_2} \cdots$ of maximal length is called a
\emph{finished derivation} of $P\cup\{Q_0\}$ when $r_1$, $r_2$, \dots are rules
from $P$ and the \emph{standardization apart} condition holds, \ie
each input rule used is variable disjoint from the initial state $Q_0$
and from the input rules used at earlier steps.

%%%%%%%%%%%%%%%%%%%%%%%%%%%%%%%%%%%%%%%%%%%%%%%%%
% Concolic Execution
%%%%%%%%%%%%%%%%%%%%%%%%%%%%%%%%%%%%%%%%%%%%%%%%%
\section{Concolic Execution}
\label{sec:concolic-exec}

In this section, we introduce a \emph{concolic execution} semantics
for CLP programs that combines both \textit{conc}rete and
symb\textit{olic} execution.
% Symbolic execution is quite natural in
% the context of a logic programming paradigm since the semantics is
% essentially the same; the only difference is that symbolic calls might
% be more general than concrete ones.
Let us now introduce some auxiliary definitions. First, we consider
unification on constraint atoms:

\begin{definition}[$\approx$, unification]
  % \approx:
  Let $C$ and $C'$ be two constraint atoms.
  If they have the same predicate symbol, \ie $C$
  has the form $\mystate{c}{p(\ol{s})}$
  and $C'$ has the form $\mystate{d}{p(\ol{t})}$
  then $C \approx C'$ denotes the formula
  $\ol{s} = \ol{t} \land c \land d$.
  Otherwise, $C \approx C'$ is $\false$.
  % Unification:
  We say that $C$ and $C'$ \emph{unify},
  or that \emph{$C$ unifies with $C'$}, when
  $C \approx C'$ is satisfiable (\ie
  $\cD \models \exists (C \approx C')$ holds).
\end{definition}
%
% \begin{lemma}\label{lemma:unify-set}
%   Let $C$ and $C'$ be some variable disjoint constraint atoms.
%   Then, $C$ and $C'$ unify if and only if
%   $\set(C)\cap\set(C')\neq\emptyset$.
% \end{lemma}
%
The following auxiliary function, $\catom$, produces a constraint atom
associated to either a state, a rule or a collection of rules.
It selects the leftmost atom together with the constraint.

\begin{definition}[$\catom$]
  For any state $Q = \mystate{c}{\ol{A_n}}$, $n > 0$,
  we let $\catom(Q)=\mystate{c}{A_1}$.
  For any rule $r=H\leftarrow c \land \ol{B}$, we let
  $\catom(r)=\mystate{c}{H}$.
  For any set of rules $\cR$ (resp. sequence of rules $\ol{r_n}$),
  we let $\catom(\cR) = \{\catom(r) \mid r\in \cR\}$
  (resp. $\catom(\ol{r_n}) = \catom(r_1), \dots, \catom(r_n)$).
\end{definition}
Function $\clauses$ is then used to determine the program rules that match
a particular state:

\begin{definition}[$\clauses$]
  Given a state $Q$ and a set of rules $P$, we let
    \[
      \clauses(Q, P) = \left\{r\in P
        \;\middle|\;
        \begin{array}{l}
          \mathit{solv}(\catom(Q)\approx\catom(r')) = \mathtt{true}\\
          \text{for some fresh copy $r'$ of $r$}
        \end{array}
        \right\}\;.
\]
\end{definition}
The following function, $\negcon$, will be essential to
guarantee that symbolic execution is sound, so that symbolic states
do not unify with more rules than expected (see
below).

\begin{definition}[$\negcon$]
  \label{def:negcon}
  Let $C:=\mystate{c}{p(\ol{s})}$ and $H:=\mystate{d}{p(\ol{t})}$ be
  some variable disjoint constraint atoms.  The constraint $\negcon(C,H)$
  denotes
  $\forall V \;(\ol{s}\neq\ol{t}\vee \neg d)$, where $V$
  denotes the set of variables occurring in $H$.

  Let $\cH:=\{\ol{H_k}\}$ be a finite set of constraint atoms that have
  the same predicate symbol as $C$ and are variable disjoint with $C$.
  Then, we let $\negcon(C,\cH) =
  \negcon(C,H_1)\land \ldots\land\negcon(C,H_k)$.
  In particular, if $\cH=\emptyset$, then we have $\negcon(C,\cH) = \true$.
\end{definition}
Given a constraint atom $C$ and a set of constraint atoms $\cH$, we
have that $C\land \negcon(C,\cH)$ does not unify with any constraint
atom in $\cH$, as expected; moreover, it is \emph{maximal} in the
sense that, for any constraint $d$ such that $C\land d$ does not unify
with any constraint atom in $\cH$, $d$ will be less general than
$\negcon(C,\cH)$ (see Propositions~\ref{prop:gamma-do-not-unify} and
\ref{prop:gamma-maximality} in \ref{appendix:concolic-exec}).

In this work, we assume that we are interested in producing test cases
that achieve a so-called full \emph{path coverage}, so that every
predicate is called in all possible ways, as explained in the introduction. More precisely, given an initial state of the form
$\mystate{\true}{p(X_1,\ldots,X_n)}$, we aim at producing test
cases that cover all \emph{feasible} subtrees of
the execution space of $\mystate{\true}{p(X_1,\ldots,X_n)}$.
We note that, since the execution space of
$\mystate{\true}{p(X_1,\ldots,X_n)}$ is typically infinite,
so is the number of feasible subtrees and, thus, the number
of test cases.
% Roughly speaking, if a
% predicate is defined by two rules, $r_1$ and $r_2$, then we want to
% produce test cases so that no rule is matched, only rule $r_1$ is
% matched, only rule $r_2$ is matched, and both rules $r_1,r_2$ are
% matched (as long as each case is feasible).
Therefore, achieving a
full path coverage is not possible
%since there are infinitely many
%possible execution paths. Therefore,
and one should introduce some
strategy to ensure the termination of concolic testing (see below).

In order to identify each derivation so that we can keep track of the
already considered derivations in the execution space, we introduce
the following notion:

\begin{definition}[trace]
  Given a rule $r$, we let $\ell(r)$ denote its label, which is unique
  in a program.
  A \emph{trace} is a sequence of rule labels. The empty trace is
  denoted by $\epsilon$. Given a trace $\pi$ and a rule label $\ell$,
  we denote by $\pi.\ell$ the concatenation of $\ell$ to the end of
  trace $\pi$.

  Given a derivation with the standard operational semantics,
  $Q_0 \lra_{r_1} Q_1 \lra_{r_2} \ldots \lra_{r_n} Q_n$, the
  associated trace is $\ell_1\ell_2\ldots\ell_n$, where $\ell(r_i)=\ell_i$,
  $i=1,\ldots,n$.
\end{definition}
In the following, we consider that states can be labelled with a
trace \ie $S_\pi$ denotes a state $S$ which is labelled with
trace $\pi$. Let us now introduce the notion of concolic state:

\begin{definition}[concolic state]
  A \emph{concolic state} has the form $\tuplec{Q\sep S_\pi}$ where
  $Q,S$ are states such that $Q\leq S$ and $\pi$
  is a trace labelling state $S$.
  Here, $Q$ is called the \emph{concrete} state of
  $\tuplec{Q \sep S_\pi}$, while $S_\pi$ is called its \emph{symbolic}
  state; we sometimes omit the trace $\pi$ from the symbolic state
  when it is not relevant.
\end{definition}
In contrast to other programming paradigms, the notion of
\emph{symbolic} execution is very natural in CLP: the structure of
both $Q$ and $S$ is the same (\ie the sequence of atoms are variants),
and the only difference (besides some labeling for symbolic states) is
that some states might be more restricted in $Q$ than in $S$.

The standard operational semantics is now extended to concolic states
as follows:

\begin{definition}[concolic execution] \label{def:concolicexec}
  Let $P$ be a program and let $\tuplec{Q\sep S_\pi}$ be a concolic
  state. Then, we have a concolic execution step
  \[
    \tuplec{Q\sep S_\pi} \stackrel{r}{\lracn}_{\pi,R_Q,R_S} \tuplec{Q'\sep
      S'_{\pi.\ell(r)}}
  \]
  if the following conditions hold:
  \begin{itemize}
  \item $\clauses(Q,P) = R_Q\neq\emptyset$, $\clauses(S,P)=R_S$,
   % $R_Q\subseteq R_S$,
  \item $\gamma=\negcon(\catom(S), \catom(R_S \setminus R_Q))$,
  \item $r\in R_Q$, $Q \lra_r Q'$ and $S\wedge\gamma \lra_r S'$.
  \end{itemize}
  Besides the applied rule, $r$, the step is labelled with the current
  trace, $\pi$, the set of rules matching the concrete state, $R_Q$, and
  the set of rules matching the symbolic state, $R_S$.\footnote{This
    information can be safely ignored in this section. It will become
    relevant in the next section in order to generate test cases.}
  The applied rule is often omitted when it is not relevant.

  A concolic state
  $\tuplec{\mystate{c}{\ol{A}}\sep\mystate{d}{\ol{B}_\pi}}$ is said to
  be \emph{successful} if $\ol{A}=\ol{B}=\epsilon$; it is said
  to be \emph{failed} if they are not empty and no derivation step is
  possible. In either case, we say that $\pi$ is the trace of the
  derivation.
  The notion of (finished) derivation is extended from the standard
  semantics in the natural way.
\end{definition}
For each concolic state $\tuplec{Q\sep S}$ in a derivation, the
symbolic component, $S$, typically unifies with more rules than the
concrete component, $Q$, since $Q$ is more restricted than $S$ (and,
thus, $R_Q \subseteq R_S$; see % Lemma~\ref{lemma:one-step-neg1}
below). However, we want the execution of the symbolic state to mimic
that of the concrete state. Therefore, both the concrete and symbolic
states can only be unfolded using a rule from $R_Q$. Furthermore, we
introduce a \emph{negative} constraint, $\gamma$, into the symbolic
state in order to avoid matching more rules than the concrete
state. For this purpose, we use function $\negcon$ introduced above.
In the remainder of the paper, we assume a fixed program $P$.

Let $\tuplec{Q\sep S}$ be a concolic state with $\clauses(Q,P) = R_Q$
and $\clauses(S,P)=R_S$. Our notion of concolic execution enjoys the
following properties (see~\ref{appendix:concolic-exec}):
\begin{itemize}
\item $Q\leq S$ implies $\clauses(Q,P) \subseteq \clauses(S,P)$.
\item $\clauses(S\wedge\gamma)=R_Q$, where
  $\gamma= \negcon(\catom(S), \catom(R_S \setminus R_Q))$. Therefore,
  $\gamma$ achieves the desired effect of preventing $S$ to unify with
  the rules in $R_S\setminus R_Q$.
\item If $\tuplec{Q\sep S} \lracn_{\pi,R_Q,R_S} \tuplec{Q'\sep S'}$,
  then $\tuplec{Q'\sep S'}$ is also a concolic state, which means that
  concolic execution is well defined in the sense that the property
  $Q\leq S$ is correctly propagated by concolic execution steps.
\end{itemize}
%
% \begin{lemma} \label{lemma:one-step-neg1} Let $Q,S$ be states with
%   $Q\leq S$. Then, $\clauses(Q,P) \subseteq \clauses(S,P)$.
% \end{lemma}
% %
% The next results show that our notion of concolic execution is well
% defined:
%
% % The following result states a fundamental property that will become
% % essential to prove that concolic execution is a conservative extension
% % of the standard operational semantics.
%
% \begin{lemma} \label{lemma:one-step-neg2} Let $\tuplec{Q\sep S}$ be a
%   concolic state with $\clauses(Q,P) = R_Q$ and
%   $\clauses(S,P)=R_S$. Then, $\clauses(S\wedge\gamma)=R_Q$, where
%   $\gamma= \negcon(\catom(S), \catom(R_S \setminus R_Q))$.
% \end{lemma}
%
% \begin{lemma} \label{lemma:one-step}
%   Let $\tuplec{Q\sep S}$ be a concolic state with
%   $\tuplec{Q\sep S} \lracn_{\pi,R_Q,R_S} \tuplec{Q'\sep S'}$.
%   Then, $\tuplec{Q'\sep S'}$ is also a concolic state.
% \end{lemma}
%
W.l.o.g., we only consider \emph{initial} concolic states of the
form $\tuplec{Q\sep S}$, where $Q= \mystate{c}{p(\ol{X})}$,
$S= \mystate{\true}{p(\ol{Y})_\epsilon}$, $Q$ and $S$ are variable
disjoint, and $\epsilon$ is the empty trace. Trivially, we have
$Q\leq S$.

In the following, we assume that all concolic execution derivations
start from an \emph{initial} concolic state, so they are well formed.

\begin{example} \label{ex:concolicexec}
  Consider the following CLP(\term) program:
  \[
    \begin{array}{l@{~~}l@{~~}l}
      \ell_1: & p(X) \leftarrow X=a. & (r_1)\\
      \ell_2: & p(s(Y)) \leftarrow \true \wedge q(Y). & (r_2)\\
      \ell_3: & q(W) \leftarrow W=a. & (r_3)\\
    \end{array}
  \]
  with rules, $r_1$, $r_2$ and $r_3$, where $\ell_1,\ell_2,\ell_3$ are
  unique identifiers for these rules.  Given the initial concolic
  state
  $\tuplec{\mystate{X=s(a)}{p(X)}\sep\mystate{\true}{p(N)_\epsilon}}$,
  we have the following concolic execution:
  %\footnote{Some constraints are slightly
   % slimplified for clarity.}
  \[
    \begin{array}{l}
      \tuplec{\blue{\mystate{X=s(a)}{p(X)}}\sep \red{\mystate{\true}{p(N)}_\epsilon}}\\
      \hspace{0ex}\lracn_{\epsilon,\{r_2\},\{r_1,r_2\}}
      \tuplec{\blue{\mystate{s(Y)=X\land X=s(a)}{q(Y)}}\sep \\
      \hspace{15ex}\sep\red{\mystate{s(Y')=N\wedge\forall X'\: (N\neq X'\vee X'\neq a)}{q(Y')}_{\ell_2}}}\\
      \hspace{0ex}\lracn_{\ell_2,\{r_3\},\{r_3\}}
      \tuplec{\blue{\mystate{W=Y\land W=a\land s(Y)=X\land X=s(a)}{\epsilon}}\\
      \hspace{14ex}\sep \red{\mystate{W'=Y'\land W'=a\land
      s(Y')=N\wedge\forall X'\: (N\neq X'\vee X'\neq a)}{\epsilon}_{\ell_2\ell_3}}}\\
    \end{array}
  \]
  In the first step, the following negative constraint is computed:
  \[
    \gamma_1 =
    \negcon(\mystate{\true}{p(N)},\{\mystate{X'=a}{p(X')}\}) = \forall
    X' (N\neq X'\vee X'\neq a)
  \]
  so that
  $
    \mystate{\true}{p(N)} \wedge \gamma_1 =
    \mystate{\forall X' (N\neq X'\vee X'\neq a)}{p(N)}
  $.
  In the second step, we have $\gamma_2 = \true$ since the matching
  rules are the same for both the concrete and symbolic states. Hence,
  no additional negative constraint is added to the symbolic state.
  %
  % Therefore, we have a successful derivation computing the constraint
  % \[
  %   W=Y\land W=a\land s(Y)=X\land X=s(a)
  % \]
  % for the concrete state (which can be simplified to $\true$ since the
  % initial concrete state has no variables), and the constraint
  % \[
  %   W'=Y'\land W'=a\land
  %     s(Y')=N\wedge\forall X'\: (N\neq X'\vee X'\neq a)
  % \]
  % for the symbolic state (which can be simplified to $N=s(a)$ when
  % restricted to the variables of the initial symbolic state, \ie
  % $N$).
  %
  The trace of the derivation is thus $\ell_2\ell_3$, \ie an
  application of rule $r_2$ followed by an application of rule $r_3$.
\end{example}
Now, we can state that concolic execution is indeed a conservative
extension of the standard operational semantics:

\begin{theorem}\label{theorem:trace}
  Let $\tuplec{Q\sep S}$ be an initial concolic state. Then, we have
  $Q \lra^\ast Q'$ iff
  $ \tuplec{Q\sep S} \lracn^\ast \tuplec{Q''\sep S'} $, where
  $Q'\equiv Q''$. Moreover, the trace of both derivations is the same.
\end{theorem}
%
% The next result is crucial to prove the soundness of our concolic
% testing approach.
%
% \begin{lemma} \label{lemma:sound} Let $\tuplec{Q\sep S}$ be a concolic
%   state with $S=\mystate{d}{\ol{A}}$. If there is a concolic execution
%   of the form
%   $\tuplec{Q\sep S} \stackrel{r}{\lracn}_{\pi,R_Q,R_S} \tuplec{Q'\sep
%     \mystate{d'\wedge d}{\ol{B}}}$, then $S\wedge d'\lra_r S''$ with
%   $S''\equiv \mystate{d'\wedge d}{\ol{B}}$.
%   %
%   Furthermore, $\rules(\catom(S\wedge d')) = R_Q$.
% \end{lemma}
%
Finally, the %following
next property states that the constraints computed for the symbolic
state ensure---when applied to the initial symbolic state---that the
standard semantics will follow the same path. Therefore, our approach
to concolic testing can be considered sound. This property did not
hold in the original approach of Mesnard et al.\ \citeyear{MPV15}, 
as explained in the introduction.

\begin{theorem}[soundness] \label{theorem:sound}
  Let $\tuplec{Q\sep S_\epsilon}$ be an initial concolic state with
  $\tuplec{Q\sep S_\epsilon} \lracn^\ast \tuplec{Q' \sep S'_\pi}$. Let
  $S = \mystate{\true}{\ol{A}}$ and $S' = \mystate{d}{\ol{B}}$. Then,
  we have $\mystate{d}{\ol{A}} \lra^\ast S''$ such that $S'\equiv S''$
  and the associated trace is $\pi$.
\end{theorem}

%%%%%%%%%%%%%%%%%%%%%%%%%%%%%%%%%%%%%%%%%%%%%%%%%
% Concolic Testing Procedure
%%%%%%%%%%%%%%%%%%%%%%%%%%%%%%%%%%%%%%%%%%%%%%%%%
\section{Concolic Testing}
\label{sec:concolic-testing-proc}

In this section, we present our concolic testing procedure, which is
based on the concolic execution semantics of the previous section.

First, we introduce a \emph{deterministic} version of concolic
execution that implements a depth-first search through the concolic
execution space (loosely inspired by the linear operational semantics
for Prolog introduced by Str\"oder et al.\  \citeyear{SESGF11}). 
This deterministic semantics
better reflects the current implementation \cite{contest} and,
moreover, allows one to keep the information that must survive
backtracking steps (\eg generated test cases and already considered
traces).

The deterministic concolic execution semantics is defined by means of
a (labelled) transition relation, $\lrac$, as shown in
Figure~\ref{fig:concolicCLP}. Now, concolic states have the form
$\tuplec{\ol{Q}\sep\ol{S}}$, where $\ol{Q}$ and $\ol{S}$ are sequences
of states (possibly labelled with a rule). Let us briefly explain the
rules:

\begin{figure}[t]
  \rule{\linewidth}{1pt}\\[-2ex]
  \[
  \hspace{-2ex}\begin{array}{r@{~}l}

    \mathsf{(backtrack)} & {\displaystyle
    \frac{\clauses(Q,P)=\emptyset\wedge \clauses(S,P)=R_S\wedge |\ol{Q}|>0 }
    {\tuplec{Q, \ol{Q} \sep S_\pi, \ol{S}}
      \lrac_{\pi,\emptyset, R_S}
      \tuplec{\ol{Q} \sep \ol{S}}}
    }\\[4ex]

    \mathsf{(next)} & {\displaystyle
    \frac{Q=\mystate{c}{\epsilon} }
    {\tuplec{Q, \ol{Q} \sep S_\pi, \ol{S}}
      \lrac_{\pi,\emptyset,\emptyset}
      \tuplec{\ol{Q} \sep \ol{S}}}
    }\\[4ex]

    \mathsf{(choice)} & {\displaystyle
    \frac{%\begin{array}{l}
      \clauses(Q,P) = \{\ol{r_n}\} \wedge n>0\wedge \clauses(S,P) = R_S
      % \\ \hspace{6ex}
             \wedge~ \gamma=\negcon(\catom(S), \catom(R_S \setminus \{\ol{r_n}\}))
             % \end{array}
             }
    {%\begin{array}{l}
       \tuplec{Q, \ol{Q} \sep S_\pi, \ol{S}}  %\\
       \lrac_{\pi,\{\ol{r_n}\},R_S}
      \tuplec{ Q^{r_1}, \ldots, Q^{r_n}, \ol{Q} \sep
      {\gamma\wedge S}_\pi^{r_1},
      \ldots,
      {\gamma\wedge S}_\pi^{r_n}, \ol{S}}
    % \end{array}
        }
    }\\[4ex]

    \mathsf{(unfold)} & {\displaystyle
    \frac{
      Q \lra_{r} R \wedge S \lra_{r} T  %\wedge R'=\mystate{d'}{\dots}\\
      }
    {\begin{array}{l}
      \tuplec{Q^{r}, \ol{Q} \sep
      {S}_\pi^{r}, \ol{S}}
      \lrac_{\pi,\emptyset,\emptyset}
      \tuplec{R, \ol{Q}
      \sep
      T_{\pi.\ell(r)}, \ol{S}}
    \end{array}}} \\
  \end{array}
  \]
  \rule{\linewidth}{1pt}
  \caption{Concolic CLP execution semantics (deterministic)} \label{fig:concolicCLP}
\end{figure}

\begin{itemize}
\item In contrast to the nondeterministic concolic execution
  semantics, unfolding is now split into two rules: \textsf{choice}
  and \textsf{unfold}. Rule \textsf{choice} creates as many copies of
  the states (both concrete and symbolic) as rules \emph{match the
    concrete state}.
  Then, rule \textsf{unfold} just unfolds the leftmost state (both
  concrete and symbolic) using the rule labeling these states.
  % Note that the current trace of the symbolic state is updated with
  % the label of the applied rule.

  Consider, for example, a concolic state
  $\tuplec{Q \sep S_\pi}$. If the nondeterministic version
  of concolic execution (cf.\ Definition~\ref{def:concolicexec})
  performs, e.g., the following step
  \[
    \tuplec{Q\sep S_\pi} \stackrel{r_1}{\lracn}_{\pi,R_Q,R_S}
    \tuplec{Q'\sep S'_{\pi.\ell(r_1)}}
  \]
%  for some $i\in\{1,\ldots,n\}$,
  with $R_Q = \{r_1,\ldots,r_n\}$,
  then the deterministic version of Figure~\ref{fig:concolicCLP}
  will perform the choice step
  \[
    \tuplec{Q\sep S_\pi} \lrac_{\pi,R_Q,R_S}
    \tuplec{Q^{r_1},\ldots,Q^{r_n}\sep \gamma\wedge
    S^{r_1}_{\pi},\ldots,\gamma\wedge
    S^{r_n}_{\pi}}
  \]
  followed by the unfolding step
  \[
  \tuplec{Q^{r_1},\ldots,Q^{r_n}\sep \gamma\wedge
    S^{r_1}_{\pi},\ldots,\gamma\wedge
    S^{r_n}_{\pi}}
    \lrac_{\pi,\{\},\{\}}
   \tuplec{Q',Q^{r_2}\ldots,Q^{r_n}\sep
    S'_{\pi.\ell(r_1)},\gamma\wedge S^{r_2}_{\pi}\ldots,\gamma\wedge
    S^{r_n}_{\pi}}
  \]
  Therefore, we reach the same states, $Q'$ and $S'$. The only
  difference is that alternative paths are stored explicitly
  in the concolic state (\ie $Q^{r_2},\ldots,Q^{r_n}$ and
  $\gamma\wedge S^{r_2}_\pi,\ldots,\gamma\wedge S^{r_n}_\pi$)
  and will be
  explored after a backtracking step (or when looking for more
  solutions, where an implicit backtracking step is performed).

\item When the concrete state does not match any rule, rule
  \textsf{backtrack} is applied. As before, the step is labeled with
  the current trace and the constraint atoms associated to the rules
  matching both the concrete (the empty set) and symbolic states. Note
  that we assume that the sequence $\ol{Q}$ is not empty; otherwise,
  the execution would be finished.

  For instance, if state $Q'$ in the example above does match
  any rule, we will perform the following backtracking step:
  \[
  \tuplec{Q',Q^{r_2}\ldots,Q^{r_n}\sep
    S'_{\pi.\ell(r_1)},\gamma\wedge S^{r_2}_{\pi}\ldots,\gamma\wedge
    S^{r_n}_{\pi}}
    \lrac_{\pi.\ell(r_1),\{\},R_{S'}}
  \tuplec{Q^{r_2}\ldots,Q^{r_n}\sep
    \gamma\wedge S^{r_2}_{\pi}\ldots,\gamma\wedge
    S^{r_n}_{\pi}}
  \]
  where $R_{S'} = \clauses(S',P)$.

\item Finally, rule \textsf{next} is applied when a solution is
  reached in order to consider alternative solutions (if any).
  In other words, our calculus explores the complete execution
  space for the initial state rather than stopping after the
  first solution is found.
\end{itemize}
The deterministic version of the concolic execution semantics
constitutes an excellent basis for implementing a concolic testing
procedure. For instance, one can consider only the computation of the
first solution by removing rule \textsf{next}.
Furthermore, one can easily guarantee termination by either limiting
the length of the considered concolic execution derivations or the
``depth'' of the search tree in order to only partially explore the
execution space.
% This is an orthogonal issue and is left as future work.

The following result stating the soundness of the deterministic
concolic execution semantics is straightforward:

\begin{theorem}
  Let $\tuplec{Q_0\sep S_0}$ be an initial concolic state. If
  $\tuplec{Q_0\sep S_0} \lrac^\ast \tuplec{Q,\ol{Q}\sep S,\ol{S}}$,
  then $\tuplec{Q_0\sep S_0}\lracn^\ast \tuplec{Q\sep S}$.
\end{theorem}
Note that the deterministic version is sound but incomplete in general
since it implements a depth-first search strategy.

% %%%%%%%%%%%%%%%%%%%%%%%%%%%%%%%%%%%%%%%%%%%%%%%%%%%
% \subsection{Auxiliary Functions and Notations}
% \label{sec:aux-notations}
%
% In imperative programming, given a conditional of the form
% $\tt if~c~then~s_1~else~s_2$, if a concolic execution enters the
% first branch (so $\tt c$ is true), one can compute an alternative
% test case by solving $\neg \tt c$, so that the resulting test case
% will now follow the \texttt{else} branch. In our CLP context, the
% problem is significantly more complex. Consider, for instance, a
% symbolic state $S$ that matches rules $r_1$ and $r_2$, while the
% corresponding concrete state only matches rule $r_1$. Here, we need
% to compute alternative test cases so that $S$ matches no rule, it
% matches rule $r_2$ only, and it matches both rules $r_1$ and $r_2$
% (assuming they are all feasible). In order to compute these
% alternative test cases, we introduce the auxiliary
% function %$\alt$ and $\alts$.

Now, we introduce a function to compute alternative test cases in a
concolic execution.
In the following definition, we consider a (symbolic) initial state
($I$), since test cases will always be particular instances of this
state, the current (symbolic) state in a derivation ($C$), the set of
atoms matching the concrete state ($\cH_Q$), and the set of atoms
matching the corresponding symbolic state ($\cH_S$).
Intuitively speaking, function $\alts$ produces alternative test cases
by restricting the initial symbolic state $I$ so that the current
symbolic state $C$ unifies with a subset $\cH^+$ of
constraint atoms from
$\cH_S$, except for the set $\cH_Q$ which was already considered.

\begin{definition}[$\alts$]
  \label{def:alts}
  Let $I,C$ be constraint atoms, with $C=\mystate{c}{B}$, and
  $\cH_Q,\cH_S$ be finite sets of constraint atoms that have the same
  predicate symbol as $C$ and all atoms are variable disjoint with
  each other. Then, %we let
  \[\def\arraystretch{1.2}
    \alts(I,C,\cH_Q,\cH_S) = \left\{
      I \land c\land \gamma \;\middle|\;
  \begin{array}{l}
    \cH^+\in \cP(\cH_S),\ \cH^+\neq \cH_Q \\
    \cH^- = \cH_S \setminus \cH^+\\
    \gamma = \negcon(C,\cH^-) \\
    c \land \gamma \text{ is satisfiable}\\
    \forall H\in\cH^+~(C\land\gamma)\approx H
  \end{array}\right\}\;.\]
  % For any $I'\in \alts(I,C,\cH_Q,\cH_S)$, we let $\cH^+(I')$
  % and $\cH^-(I')$ denote the sets $\cH^+$ and $\cH^-$ that
  % correspond to $I'$.
\end{definition}
%
%
% \begin{definition}[$\alts$]
%   \label{def:alts}
%   Let $I = \mystate{c_0}{q(\ol{t})}$ be a constraint atom and
%   $\cH,\cH'$ be finite sets of constraint atoms that have the
%   same predicate symbol as $C$. Then, we let
%   \[\def\arraystretch{1.2}
%     \alts(I,C,\cH,\cH') = \left\{
%       I \land \alt(C,\cH_1,\cH_2) \;\middle|\;
%   \begin{array}{l}
%     \cH_1\in \cP(\cH'),\ \cH_1\neq \cH \\
%     \cH_2 = \cH' \setminus \cH_1\\
%     c_0 \land \alt(C,\cH_1,\cH_2) \text{ is satisfiable}
%   \end{array}\right\}\;.\]
% \end{definition}
%
%
% In the following examples, we fix a constraint atom
% $C=\mystate{c}{p(\ol{s})}$ and two finite sets $\cH^+=\{\ol{H_n}\}$
% and $\cH^-=\{\ol{H'_m}\}$ of constraint atoms that have the same
% predicate symbol as $C$, and are all variable disjoint with each other.
% Furthermore, for the sake of readability,
% below we write simply $\gamma$ instead of $\negcon(C,\cH^-)$.

% \begin{definition}[$\alt$]
%   \label{def:alt}
%   We let $\alt(C,\cH^+,\cH^-)$ denote the constraint
%   $(C \approx H_1) \land \cdots \land (C \approx H_n) \land
%   \lnot (C \approx H'_1) \land \cdots \land
%   \lnot (C \approx H'_m)$.
% \end{definition}

% For the sake of readability, and as $C$, $\cH^+$ and $\cH^-$ are fixed
% in this section, below we write $\alt$ instead of $\alt(C,\cH^+,\cH^-)$.

\begin{example}[CLP(\term)] \label{ex-alt-term} Let us consider
  the call $\alts(I,C,\{H_2\},\{H_1,H_2\})$, where
  $I := \mystate{true}{p(W)}$, $C:=\mystate{c}{q(N)}$ with $c:=(W=N)$,
  $H_1 := \mystate{X=a}{q(X)}$, and $H_2 :=
  \mystate{\true}{q(s(M))}$. For brevity, we remove the occurrences of
  $\true$ in the formul\ae{} below.

  Let us consider the case $\cH^+ := \{H_1\}$ and $\cH^- := \{H_2\}$.
  Then we have $\gamma = \negcon(C,\cH^-)$ $ = \forall M (N\neq
  s(M))$. As $\cD \models_v (c\land\gamma)$ holds for any valuation
  $v$ with $\{W\mapsto a,N\mapsto a\}\subseteq v$, $c\land\gamma$ is
  satisfiable.

  Now, we should check that $C\land\gamma\approx H_1$ holds.
  Since
  $C\land\gamma = \mystate{W=N\land \forall M\; (N\neq s(M))}{q(N)}$
  and
  $\mathit{solv}(N=X\land X=a\land W=N\land \forall M\; (N\neq
  s(M)))=\true$, it holds. Therefore, we have
  $I\land c\land \gamma \in \alts(I,C,\{H_2\},\{H_1,H_2\})$, \ie we
  produce the
  %following
  state:
    $
      \mystate{W=N\land \forall M\; (N\neq s(M))}{p(W)}
    $
    which could be simplified to
    $\mystate{\forall M\; W\neq s(M)}{p(W)}$.
\end{example}

\begin{example}[CLP($\cN)$]
  \label{ex-alt-Qlin}
  Let us consider the call $\alts(I,C,\{H_1\},\{H_1,H_2,H_3\})$, where
  $I := \mystate{\true}{p(W)}$, $C := \mystate{c}{q(X)}$,
  %$\cH_Q = \{H_1\}$ and $\cH_S = \{H_1,H_2,H_3\}$, with
  $c:=(W=X\land X \le 10)$, $H_1 := \mystate{Y\le 2}{q(Y)}$,
  $H_2 := \mystate{8\le Z\le 10}{q(Z)}$, and
  $H_3 := \mystate{T < 5}{q(T)}$.

  Let us consider the case
  $\cH^+ = \{H_1,H_2\}$ and $\cH^- = \{H_3\}$.
  First, we should compute $\gamma = \negcon(C,\cH^-)$, \ie
  $\forall T\; (X\neq T\vee 5\le T)$, which can be simplified to
  $\gamma = (5\le X)$. So,
  $c\land\gamma=(W=X\land X \le 10 \land 5\le X)$ can be simplified to
  $c\land\gamma=(W=X\land 5\le X \le 10)$, which is clearly
  satisfiable. Now, we should check that
  $C\land \gamma = \mystate{W=X\land 5\le X \le 10}{q(X)}$ unifies
  with both $H_1$ and $H_2$ in order to produce an element of
  $\alts(I,C,\cH_Q,\cH_S)$:
  \begin{itemize}
  \item $C\land\gamma \approx H_1$. In this case, we have
    $\mathit{solv}(X=Y\land  Y\le 2\land W=X\land 5\le X
    \le 10) = \false$.
  \item $C\land\gamma \approx H_2$. In this case, we have
    $\mathit{solv}(X=Z\land 8\le Z\le 10\land W=X\land 5\le X
    \le 10) = \true$ (consider, \eg any valuation $v$ with
    $\{X\mapsto 9,Z\mapsto 9, W\mapsto 9\}\subseteq v$).
  \end{itemize}
  Therefore, this case is not feasible and no new test case is
  produced for it.

 %  Let us now consider instead the case $\cH^+ = \{H_1,H_3\}$ and
 %  $\cH^- = \{H_2\}$.
 %  %
 %  First, we should compute $\gamma' = \negcon(C,\cH^-)$, \ie
 %  $\forall Z\; (X\neq Z\vee Z<8\vee Z>10)$, which can be simplified to
 %  $\gamma' = (X<8\vee X>10)$. So,
 %  $c\land\gamma'=(W=X\land  X \le 10 \land (X<8\vee X>10))$
 %  can be simplified to $c\land\gamma'=(W=X\land  X < 8)$,
 %  which is clearly satisfiable. Now, we should
 %  check that $C\land \gamma'$ unifies with both $H_1$ and $H_3$ in
 %  order to produce an element of $\alts(I,C,\cH_Q,\cH_S)$:
 %  \begin{itemize}
 %  \item $C\land\gamma' \approx H_1$. In this case, we have
 %    $\mathit{solv}(X=Y\land  Y\le 2\land W=X\land  X < 8)
 %    =\true$ (consider, \eg any valuation $v$ with
 %    $\{X\mapsto 1,Y\mapsto 1,W\mapsto 1\}\subseteq v$).

 % \item $C\land\gamma' \approx H_3$. In this case, we have
 %    $\mathit{solv}(X=T\land T<5\land W=X\land  X < 8) =\true$
 %    (consider, \eg any valuation $v$ with
 %    $\{X\mapsto 4,T\mapsto 4,W\mapsto 4\}\subseteq v$).
 %  \end{itemize}
 %  Therefore, we have
 %  $I\land c \land \gamma'\in\alts(I,C,\cH_S,\cH_Q)$, \ie
 %  we produce the state:
 %  \[
 %    \mystate{W=X\land  X < 8}{p(W)}   %%\;.
 %  \]
 %  which can be simplified to $\mystate{W < 8}{p(W)}$.
\end{example}

%%%%%%%%%%%%%%%%%%%%%%%%%%%%%%%%%%%%%%%%%%%%%%%%%%%%%%%%%%%%%%%%%%
\subsection{A Concolic Testing Procedure}

\begin{figure}[t]
  \rule{\linewidth}{1pt}\\[-2ex]
  \[
  \begin{array}{r@{~}l}
   (\mathsf{skip}) & {\displaystyle
    \frac{
    \cC \lrac_{\pi,R_Q,R_S} \cC'
    \wedge (\pi\in\TR \vee R_S = \emptyset)
      }
    {\begin{array}{l}
       (\PTC,\TC,\TR,I,\cC) \leadsto
       (\PTC,\TC,\TR\cup\{\pi\},I,\cC')
     \end{array}}
    } \\[4ex]

    % {\displaystyle
    % \frac{\pi\not\in\TR\wedge
    % \cC \lrac_{\pi,\emptyset,\emptyset} \cC'
    %   }
    % {\begin{array}{l}
    %    (\PTC,\TC,\TR,S_0,\cC) \leadsto
    %    (\PTC,\TC,\TR\cup\{\pi\},S_0,\cC')
    %  \end{array}}
    % } \\[4ex]

     (\mathsf{alts}) &   {\displaystyle
    \frac{
    \cC \lrac_{\pi,R_Q,R_S} \cC' \wedge R_S\neq\emptyset
    \wedge \pi\not\in\TR
      }
    {%\begin{array}{l}
       (\PTC,\TC,\TR,I,\cC) %\\
        %\hspace{10ex}
        \leadsto
       (\PTC\cup \alts(I,\catom(\cC),\catom(R_Q),\catom(R_S)),\TC,\TR\cup\{\pi\},I,\cC')
       % \end{array}
        }
    }\\[3ex]

        (\mathsf{restart}) & {\displaystyle
    \frac{
    \cC \not\lrac
      }
    {%\begin{array}{l}
       (\PTC\cup\{\mystate{c}{p(\ol{X})}\},\TC,\TR,I,\cC) %\\
       % \hspace{10ex}
        \leadsto
       (\PTC,\TC\cup\{\mystate{c}{p(\ol{X})}\},\TR,
       I,\tuplec{\mystate{c}{p(\ol{X})}\sep I_\epsilon})
       % \end{array}
        }
    }
  \end{array}
  \]
  \rule{\linewidth}{1pt}
  \caption{Concolic testing} \label{fig:concolictesting}
\end{figure}

Now, we consider a concolic testing procedure that aims at achieving a
full path coverage.
Let us first informally explain the concolic testing procedure.
The process starts with some arbitrary test case \ie an
initial concrete state of the form
$\mystate{c}{p(\ol{X})}$. Then, concolic testing proceeds iteratively
as follows:
\begin{enumerate}
  \item First, we form the initial concolic state
  $\tuplec{\mystate{c}{p(\ol{X})}\sep\mystate{\true}{p(\ol{Y})}}$
  and apply the rules of concolic execution
  (Figure~\ref{fig:concolicCLP}) as much
  as possible (or up to a number of steps or a time bound, in order
  to ensure the termination of the process).

  \item Now, for each choice or backtrack steps in this derivation,
  we use function $\alts$ to compute alternative test cases that
  will produce a different execution tree. Moreover, we keep track
  of the traces where alternative test cases have been produced
  in order to avoid producing the same alternative test cases
  once and again.

  \item When all alternative test cases for the considered concolic
  execution have been produced,
  we go back to step (1) above and consider any of the pending
  test cases produced in the previous step. The iterative algorithm
  terminates when all pending test cases have been considered
  and, moreover, no new test cases are produced.
\end{enumerate}
In order to formalise the above process,
we introduce \emph{configurations} of the
form $(\PTC,\TC,\TR,I,\cC)$, where $\PTC$ is the set of \emph{pending
  test cases} (test cases that have not been explored yet), $\TC$ is
the set of test cases already explored, $\TR$ is the set of execution
traces already considered, $I$ is the initial symbolic state, and
$\cC$ is a concolic state. The rules of the concolic testing procedure
are shown in Figure~\ref{fig:concolictesting}.

Concolic testing starts with an arbitrary concrete state, say
$\mystate{c}{p(\ol{X})}$. Then, we form the initial configuration
\[
  (\emptyset,\{\mystate{c}{p(\ol{X})}\},\emptyset,
  \mystate{\true}{p(\ol{Y})},\tuplec{\mystate{c}{p(\ol{X})}\sep
    \mystate{\true}{p(\ol{Y})}_\epsilon})
\]
where $\ol{Y}$ are fresh variables,
and apply the rules of Figure~\ref{fig:concolictesting} until no rule
is applicable. The second component of the last configuration will
contain the generated test cases. Let us briefly explain the rules of
the concolic testing procedure:
\begin{itemize}
\item Rule \textsf{skip} applies when either the trace of the
  current state, $\pi$, is already visited or the set of rules
  matching the symbolic state is empty. The second situation happens
  in rules \textsf{next} and \textsf{unfold} of the concolic execution
  semantics, and also when applying rule \textsf{backtrack} but no
  rule matches the symbolic state. In this case, we simply update the
  concolic state and the set of considered traces (if any), but no new
  alternative test cases are produced.

\item Rule \textsf{alts} applies when the current trace, $\pi$, has
  not been considered yet and, moreover, the set of rules matching
  the symbolic state is not empty. This situation happens when applying
  rules \textsf{backtrack} or \textsf{choice} for the first time. In
  this case, we update the set of pending test cases using the
  auxiliary function $\alts$. Here, we let
  $\catom(\cC)=\catom(S)$ when $\cC=\tuplec{Q,\ol{Q}\sep S,\ol{S}}$.

\item Finally, rule \textsf{restart} applies when the concolic
  execution semantics cannot proceed. In this case, we restart the
  process with a new concrete state from the set of pending test
  cases.
\end{itemize}
The procedure terminates when the set of pending tests cases is
empty.\footnote{Note that termination of concolic testing is ensured
  when concolic execution terminates; see the previous section for
  some possible strategies.} Then, the generated test cases can be
    found in the second component of the configuration.
    A detailed example
    can be found in~\ref{apendix:examples}.

We note that, in general, concolic testing might produce
\emph{nonterminating} test cases. Here, one could use the output of
some termination analysis to further restrict test cases in order to
guarantee terminating computations (\eg requiring ground arguments
or fixed variables). This is an orthogonal issue that constitutes an
interesting topic for further research.

%%%%%%%%%%%%%%%%%%%%%%%%%%%%%%%%%%%%%%%%%%%%%%%%%%%
\subsection{Connections with the Constraint Selective Unification Problem}
\label{sec:connections-csup}

Here, we fix a constraint atom $C=\mystate{c}{p(\ol{s})}$
with $c$ satisfiable  and two finite sets $\cH^+$ and $\cH^-$ of
constraint atoms. We assume that all constraint atoms are variable
disjoint with each other and that $C$ unifies with any constraint atom from
$\cH^+ \cup \cH^-$. We recall the definition of a
\emph{constraint selective unification problem}
minus its groundness condition \cite{MPV17}.

\begin{definition}[Constraint Selective Unification Problem, $\cP$]
  \label{def:constraint_sup}
  The \emph{constraint selective unification problem} for $C$ with
  respect to $\cH^+$ and $\cH^-$ consists in determining whether the
  following set of constraint atoms is empty:
  \[\def\arraystretch{1.2}
  \cP(C,\cH^+,\cH^-) = \left\{C \land d \;\middle|\;
  \begin{array}{l}
    c \land d \text{ is satisfiable}\\
    \text{$C\land d$ is variable disjoint with $\cH^+ \cup \cH^-$} \\
    \forall H\in\cH^+: \text{$C\land d$ and $H$ unify} \\
    \forall H\in\cH^-: \text{$C\land d$ and $H$ do not unify}
  \end{array}\right\}\;.
  \]
\end{definition}
For brevity, and as $C$, $\cH^+$ and $\cH^-$ are fixed in this section,
below we write $\cP$ instead of $\cP(C,\cH^+,\cH^-)$ and we let
$\gamma = \negcon(C,\cH^-)$.

\begin{proposition}\label{prop:alts-in-p}
\begin{enumerate}
\item   Suppose that $c\land\gamma$ is satisfiable
  and $C\land\gamma$ unifies
  with each element of $\cH^+$. Then, we have $C\land\gamma \equiv C'$
  for some $C'\in\cP$.
%\end{proposition}
%
%Note that $C\land\gamma\not\in\cP$.
%Indeed, as $\gamma$ contains
%the variables of $\cH^-$, we have that $C\land\gamma$ is not variable
%disjoint with $\cH^+ \cup \cH^-$, so the condition
%``$C\land d$ is variable disjoint with $\cH^+ \cup \cH^-$"
%in Def.~\ref{def:constraint_sup} does not hold for $C\land\gamma$.
%
%\begin{proposition}\label{prop:p-in-alts}
\item For each $C'\in\cP$ we have $C' \leq (C\land\gamma)$.
%\end{proposition}
%
%Therefore, $C\land\gamma$ is maximal.

%\begin{proposition}\label{prop:if-P-neq-emptyset}
\item If $\cP\neq\emptyset$ then $C\land\gamma$ unifies
  with each constraint atom in $\cH^+$.
\end{enumerate}
\end{proposition}
Below, we naturally let $\set(\cP)=\cup_{C'\in \cP}\set(C')$.

\begin{theorem}\label{thm:if-C-and-gamma-unifies}
  If $C\land\gamma$ unifies with each element of $\cH^+$
  then $\set(C\land\gamma) = \set(\cP)$.
\end{theorem}

\begin{corollary}\label{cor:decidability}
  The \emph{constraint selective unification problem} for $C$ with
  respect to $\cH^+$ and $\cH^-$ is decidable.
\end{corollary}

%%%%%%%%%%%%%%%%%%%%%%%%%%%%%%%%%%%%%%%%%%%%%%%%%%%%%%
\section{Related Work}

Concolic testing was originally introduced in the context of
imperative programming languages \cite{GKS05,SMA05} and, then,
extended to a concurrent language like Java by 
Sen and Agha \citeyear{SenA06}.
To the best of our knowledge, 
the first work that considered concolic execution in the context 
of a nondeterministic, logic programming language was 
that of Vidal \citeyear{Vid14}, 
where some preliminary ideas were introduced. 
However, the paper presented no formal
results nor an implementation of the technique. Later,
a more mature approach was proposed by Mesnard et at.\
\citeyear{MPV15}, where the formal concept of a selective unification problem,
together with a correct, terminating but incomplete algorithm to solve
it, were introduced. The soundness of concolic execution itself was not
considered and, indeed, it was not sound, as illustrated in
Section~\ref{sec:intro}.

A publicly available proof-of-concept implementation of a concolic
testing tool for (pure) Prolog has been developed: \textsf{contest}
\cite{contest}.  Our present paper generalizes the approach to CLP
with first order constraints, which provides some crucial help thanks
to negative constraints to prove the operational soundness of our
concolic scheme: a generated test case will indeed follow the intended
execution path.

Mesnard et al.\ \citeyear{MPV17} showed that requiring a traditional
constraint solver (\ie a decision procedure for existentially
quantified conjunction of atomic constraints) is not enough to decide
the constraint selective unification problem (CSUP).  Indeed, we
presented a CLP instance based on the theory of arrays where we proved that the
CSUP is undecidable.  Then we showed that assuming variable
elimination together with a traditional constraint solver is enough to
decide the CSUP.
%We gave a generic algorithm, which we instantiated for CLP($\cQ$).
Of course, a constraint domain with both a
traditional constraint solver and a variable elimination algorithm is
decidable.  But solving the CSUP without variable elimination was an
open question in the paper by Mesnard et al.\ \citeyear{MPV17}.  
In the present paper, we have
presented a more general approach that can solve the CSUP for decidable constraint domains without
variable elimination.  CLP($\cN$) and CLP(\term) are two such
constraint domains.

% A concolic testing tool for (pure) Prolog has been developed
% \cite{contest}. More recently, a prototype implementation that
% includes an interface with an SMT solver has been undertaken
% \cite{FMPPVV20}, with the aim of avoiding the issue shown in the
% previous example.

In turn, Fortz et al.\ \citeyear{FMPPVV20} %without any strong theoretical results,
essentially showed that one could rely on an SMT solver to implement
a concolic testing tool for Prolog. %logic programming.
The paper is focused on designing a more efficient alternative
implementation of \textsf{contest}, as well as trying to avoid the
unsoundness of the original approach by Mesnard et al.\  
\citeyear{MPV15}. Unfortunately, the ideas in this paper
are preliminary and it does not provide any theoretical
result. Moreover, it only considers pure logic programs, so even if
negative constraints are used during concolic testing, they cannot be
represented in the generated test cases.

Finally, one can also find some similarities
with an approach proposed by Leuschel and De~Schreye
\citeyear{LS98}
in the context of partial deduction~\cite{LS91}.
In particular,
the partial deduction algorithm of Gallagher and
Bruynooghe \citeyear{GB91} introduced the
use of abstract interpretation based on so-called \emph{characteristic
paths} which, roughly speaking, described the deterministic part
of the unfolding of an atom. The authors aimed at preserving
these characteristic paths when computing resultants and
their (most specific) generalisation.
However, as noted by Leuschel and De~Schreye \citeyear{LS98},
this property does not hold, since the generated resultants
are sometimes less deterministic than the original rules.
In order to overcome this problem, Leuschel and De~Schreye
\citeyear{LS98} extended the framework
of Gallagher and
Bruynooghe \citeyear{GB91} to a constraint setting and, moreover,
introduce some \emph{pruning} constraints to avoid matching
more rules than expected.
Although in a different context, this is essentially the
same solution that we have proposed in this
paper in order to overcome the limitations of 
Mesnard et al.\ \citeyear{MPV15}.

%Concolic execution and partial deduction seem related:
%one can view concrete execution as unfolding an atom $A$, and
%the symbolic part as the most general abstraction of $A$
%which produces the same computation tree.
%The unsoundness of the approach of~\cite{MPV15} reminds of an
%unsoundness in an abstraction operator for partial deduction
%introduced in~\cite{GB91}. This operator was supposed to preserve
%the computation trees (called \emph{characteristic trees} in
%this paper). However, there was an error in one of the lemmas.
%The error was uncovered and then rectified in~\cite{LS98}, leading
%to a technique which uses negative constraints (as we do herein) to
%prune the tree of the abstraction into the right shape.

%%%%%%%%%%%%%%%%%%%%%%%%%%%%%%%%%%%%%%%%%%%%%%%%%%%%%%%%%%%%%%%%%%%%%%%%%
\section{Conclusion and Future Work}

In this paper, we have extended concolic testing to CLP. Thanks to the availability of
negative constraints, we have formulated and proved a precise operational soundness
criteria. Moreover, we have proved that for decidable constraint domains,
the selective unification problem is decidable too. Hence, our approach
constitutes an excellent basis for designing  a powerful concolic
testing  tool for CLP programs.

For future work, we consider the definition of a post-processing that takes the
generated test cases, and further restricts them (if needed) in
order to ensure that their execution is always terminating. For
this purpose, we may consider the output of some termination
analysis for CLP programs. Moreover, we plan to deal with a subset of built-ins in order to
cope with practical issues. Finally, we will explore the use of
types  (as defined in Typed Prolog~\cite{SCWD08} or
Mercury~\cite{SHC96}) to further restrict the possible values a variable
can take when generating test cases.

% ---- Bibliography ----
\bibliographystyle{acmtrans}
%\bibliography{biblio}

% ---- Appendix ----

\newpage
\appendix

%%%%%%%%%%%%%%%%%%%%%%%%%%%%%%%%%%%%%%%%%%%%%%%%%
% Concolic Execution
%%%%%%%%%%%%%%%%%%%%%%%%%%%%%%%%%%%%%%%%%%%%%%%%%
\section{Proofs for Section~\ref{sec:concolic-exec}}
\label{appendix:concolic-exec}

In this section, we provide proofs for the technical results of
Section~\ref{sec:concolic-exec} as well some additional properties.
Let us first consider the following lemma:
\begin{lemma}\label{lemma:unify-set}
  Let $C$ and $C'$ be some variable disjoint constraint atoms.
  Then, $C$ and $C'$ unify if and only if
  $\set(C)\cap\set(C')\neq\emptyset$.
\end{lemma}
\begin{proof}
  Let $C$ and $C'$ be two constraint atoms.
  \begin{itemize}
    \item Suppose that $C$ and $C'$ unify. Then, they have
    the same predicate symbol \ie $C$ has the form $\mystate{c}{p(\ol{s})}$
    and $C'$ has the form $\mystate{d}{p(\ol{t})}$.
    Moreover, $C \approx C'$ is satisfiable \ie there exists a
    valuation $v$ such that $\cD\models_v(\ol{s}=\ol{t} \land c \land d)$.
    Note that $p([\ol{s}]_v)\in\set(C)$, $p([\ol{t}]_v)\in\set(C')$ and
    $[\ol{s}]_v=[\ol{t}]_v$. So, $\set(C) \cap \set(C') \neq \emptyset$.

    \item Suppose that $\set(C)\cap\set(C')\neq\emptyset$.
    Then necessarily $C$ and $C'$ have the same predicate symbol \ie $C$
    has the form $\mystate{c}{p(\ol{s})}$ and $C'$ has the form
    $\mystate{d}{p(\ol{t})}$.
    Let $p(\ol{a})\in \set(C)\cap\set(C')$. Then there exists a valuation
    $v_1$ such that $\cD\models_{v_1} c$ and $\ol{a}=[\ol{s}]_{v_1}$ and a
    valuation $v_2$ such that $\cD\models_{v_2} d$ and
    $\ol{a}=[\ol{t}]_{v_2}$. Hence, as $C$ and $C'$ are variable disjoint,
    there exists a valuation $v$ such that $v(V)=v_1(V)$ for all variable $V$
    occurring in $C$ and $v(V)=v_2(V)$ for all variable $V$ occurring in
    $C'$. Then, we have
    $[\ol{s}]_v=[\ol{s}]_{v_1}=\ol{a}$,
    $[\ol{t}]_v=[\ol{t}]_{v_2}=\ol{a}$,
    $[c]_{v}=[c]_{v_1}=1$ and $[d]_{v}=[d]_{v_2}=1$. Consequently,
    we have $\cD\models_v(\ol{s}=\ol{t} \land c \land d)$ \ie
    $C\approx C'$ is satisfiable. So we have proved that
    $C$ and $C'$ unify.
  \end{itemize}
\end{proof}
Now, we prove the following proposition, which states an essential
property of %$function 
$\negcon$:

\begin{proposition}
  \label{prop:gamma-do-not-unify}
  Let $C$ be a constraint atom and $\cH$ be a finite
  set of constraint atoms that have the same predicate symbol as $C$
  and are variable disjoint with $C$.
  Then, $C\land \negcon(C,\cH)$ does not unify with any constraint atom
  in $\cH$.
\end{proposition}

\begin{proof} %[Proof of Proposition~\ref{prop:gamma-do-not-unify}]
  If $\cH$ is empty, then the result holds vacuously.
  Now suppose that $\cH$ is not empty and let $H=\mystate{d}{p(\ol{t})}$
  be a constraint atom in $\cH$.
  Suppose that $C=\mystate{c}{p(\ol{s})}$.
  For the sake of readability, let $\gamma=\negcon(C,\cH)$.
  Then, $(C\land\gamma)\approx H$ is the formula
  $\left(\ol{s}=\ol{t}\land c \land \gamma \land d\right)$.
  Therefore, $(C\land\gamma)\approx H$ contains the conjunct
  $\ol{s}=\ol{t} \land d$ together with the conjunct
  $\negcon(C,H) = \forall V (\ol{s}\neq\ol{t} \lor \lnot d) =
  \forall V \lnot(\ol{s}=\ol{t} \land d)$, where $V$
  denotes the set of variables occurring in $H$.
  So $(C\land\gamma)\approx H$ is not satisfiable
  \ie $C\land\gamma$ does not unify with $H$.
\end{proof}
The next proposition states that $\negcon(C,\cH)$ is maximal, in the
sense that it captures all the constraints that make $C$ non-unifiable
with the elements of $\cH$.
\begin{proposition}[Maximality of $\negcon$]
  \label{prop:gamma-maximality}
  Let $C$ be a constraint atom, $d$ be a constraint and $\cH$ be a finite
  set of constraint atoms that have the same predicate symbol as $C$ and
  are variable disjoint with $C\land d$.
  If $C\land d$ does not unify with any constraint atom in $\cH$
  then we have $\left(C\land d\right) \leq \left(C\land \negcon(C,\cH)\right)$.
\end{proposition}

\begin{proof} %[Proof of Proposition~\ref{prop:gamma-maximality}]
  For the sake of readability, we let $\gamma=\negcon(C,\cH)$.
  Suppose that $C=\mystate{c}{p(\ol{s})}$, $\cH=\{\ol{H_k}\}$
  and $C\land d$ does not unify with any constraint atom in $\cH$.
  \begin{itemize}
    \item Suppose that $\set(C\land d)$ is empty. Then, the result
    trivially holds.
    \item Suppose that $\set(C\land d)$ is not empty.
    Let $p(\ol{a})\in\set(C\land d)$.
    Then, there exists a valuation $v$ such that $\cD\models_v(c\land d)$
    and $\ol{a}=[\ol{s}]_v$.

    For any $H_i=\mystate{e}{p(\ol{t})}\in\cH$, as $C\land d$ does not unify
    with $H_i$, we have $\set(C\land d)\cap\set(H_i)=\emptyset$ by
    Lemma~\ref{lemma:unify-set}. So, $p(\ol{a})\not\in\set(H_i)$ \ie
    $\cD\models_v \forall V \lnot(\ol{s}=\ol{t} \land e)$
    where $V$ denotes the set of variables occurring in $H_i$.
    Hence, $\cD\models_v \negcon(C,H_i)$.

    Consequently, we have
    $\cD\models_v \negcon(C,H_1)\land \ldots\land\negcon(C,H_k)$
    \ie $\cD\models_v \gamma$. Moreover, as
    $\cD\models_v(c\land d)$ then in particular $\cD\models_v c$.
    Therefore, we have $\cD\models_v c \land \gamma$.
    So, $[\ol{s}]_v=\ol{a} \in \set(C\land\gamma)$.

    We have then proved that $\set(C\land d) \subseteq
    \set(C\land \gamma)$. Hence the result.
  \end{itemize}
\end{proof}
The next lemma states a basic property of states:

\begin{lemma} \label{lemma:one-step-neg1} Let $Q,S$ be states with
  $Q\leq S$. Then, $\clauses(Q,P) \subseteq \clauses(S,P)$.
\end{lemma}

\begin{proof} %[Proof of Lemma~\ref{lemma:one-step-neg1}]
  If $\clauses(Q,P)=\emptyset$ then the result trivially holds.
  Now suppose that $\clauses(Q,P)\neq\emptyset$ and let
  $r\in\clauses(Q,P)$. Then, for some fresh copy $r'$ of $r$, we have
  $\mathit{solv}(C \approx R) = \mathtt{true}$ where $C=\catom(Q)$ and
  $R=\catom(r')$.
  So, by Lemma~\ref{lemma:unify-set},
  we have $\set(C)\cap\set(R)\neq\emptyset$.
  Therefore, as $Q\leq S$, we have $\set(C')\cap\set(R)\neq\emptyset$
  where $C'=\catom(S)$. Hence, again by Lemma~\ref{lemma:unify-set},
  we have $\mathit{solv}(C' \approx R) = \mathtt{true}$ \ie
  $r\in\clauses(S,P)$. Here, we assumed the same variant $r'$ of $r$
  for simplicity.
\end{proof}
The next results show that our notion of concolic execution is well
defined:

\begin{lemma} \label{lemma:one-step-neg2} Let $\tuplec{Q\sep S}$ be a
  concolic state with $\clauses(Q,P) = R_Q$ and
  $\clauses(S,P)=R_S$. Then, $\clauses(S\wedge\gamma)=R_Q$, where
  $\gamma= \negcon(\catom(S), \catom(R_S \setminus R_Q))$.
\end{lemma}

\begin{proof} %[Proof of Lemma~\ref{lemma:one-step-neg2}]
  Since $\tuplec{Q\sep S}$ is a concolic state, we have $Q\leq S$. By
  Lemma~\ref{lemma:one-step-neg1}, we have
  $\rules(Q,P)\subseteq\rules(S,P)$.
  By Proposition~\ref{prop:gamma-do-not-unify}, we have that
  $\rules(S\land\gamma,P)\subseteq\rules(Q,P)$.
  Now, we only need to prove that
  $\rules(Q,P)\subseteq \rules(S\land\gamma,P)$.
  If $\clauses(Q,P)=\emptyset$ then the result trivially holds.  Now
  suppose that $\clauses(Q,P)\neq\emptyset$ and let
  $r\in\clauses(Q,P)$. Then, for some fresh copy $r'$ of $r$, we have
  $\mathit{solv}(C \approx R) = \mathtt{true}$ where $C=\catom(Q)$ and
  $R=\catom(r')$.
  Let $C' = \catom(S)$. Since $Q\leq S$, we have that $C = C'\land d$
  for some constraint $d$. By Proposition~\ref{prop:gamma-maximality},
  we have $C= \left(C'\land d\right) \leq \left(C'\land \gamma\right)$.
  By Lemma~\ref{lemma:unify-set}, we have
  $\set(C)\cap\set(R)\neq\emptyset$.  Therefore, as
  $C \leq \left(C'\land \gamma\right)$, we have
  $\set(C'\land\gamma)\cap\set(R)\neq\emptyset$. Hence, again by
  Lemma~\ref{lemma:unify-set}, we have
  $\mathit{solv}(C'\land\gamma \approx R) = \mathtt{true}$ \ie
  $r\in\clauses(S\land\gamma,P)$. Here, we assumed the same variant
  $r'$ of $r$ for simplicity.
\end{proof}

\begin{lemma} \label{lemma:one-step}
  Let $\tuplec{Q\sep S}$ be a concolic state with
  $\tuplec{Q\sep S} \lracn_{\pi,R_Q,R_S} \tuplec{Q'\sep S'}$.
  Then, $\tuplec{Q'\sep S'}$ is also a concolic state.
\end{lemma}

\begin{proof} %[Proof of Lemma~\ref{lemma:one-step}]
  Let $Q=\mystate{c}{p(\ol{u}),\ol{A}}$ and
  $S=\mystate{c'}{p(\ol{u'}),\ol{A'}}$.  Since $\tuplec{Q\sep S}$ is a
  concolic state, we have $Q\leq S$ and, moreover, $p(\ol{u}),\ol{A}$
  and $p(\ol{u'}),\ol{A'}$ are variants.
  By Lemma~\ref{lemma:one-step-neg2}, we have
  $\rules(Q,P)=\rules(S\land\gamma,P)$.  Let $r
  \in\clauses(Q,P)$. Then, for some fresh copy
  $r' = (p(\ol{s}) \leftarrow d\land\ol{B})$ of $r$, we have
  $\mathit{solv}(\ol{s}=\ol{u}\land d\land c) = \true$.
  Since $Q\leq S$, we have that
  $\mystate{c}{p(\ol{u})} = \mystate{c'}{p(\ol{u'})} \land d'$ for
  some constraint $d'$. By Proposition~\ref{prop:gamma-maximality}, we
  have
  $\mystate{c}{p(\ol{u})} = (\mystate{c'}{p(\ol{u'})} \land d') \leq
  (\mystate{c'}{p(\ol{u'})} \land \gamma) = \mystate{c'\land
    \gamma}{p(\ol{u'})}$.
  By definition of concolic execution, we have
  $Q' = \mystate{\ol{s}=\ol{u}\land d\land c}{\ol{B},\ol{A}}$ and
  $S' = \mystate{\ol{s}=\ol{u'}\land d\land c'}{\ol{B},\ol{A'}}$ (we
  consider the same renaming of $r$ for simplicity).
  Therefore, the claim follows from
  $\mystate{c}{p(\ol{u})} \leq \mystate{c'\land \gamma}{p(\ol{u'})}$
\end{proof}
Finally, we can prove the main results of this section:

%%\begin{theorem}\label{theorem:trace}
  \mbox{}\\
  \textit{Theorem~\ref{theorem:trace}}\\
  Let $\tuplec{Q\sep S}$ be an initial concolic state. Then, we have
  $Q \lra^\ast Q'$ iff
  $ \tuplec{Q\sep S} \lracn^\ast \tuplec{Q''\sep S'} $, where
  $Q'\equiv Q''$. Moreover, the trace of both derivations is the same.
%\end{theorem}

\begin{proof} %[Proof of Theorem~\ref{theorem:trace}]
  The claim follows by a simple induction on the length of the
  considered derivation since concolic execution boils down to the
  standard operational semantics regarding concrete states, the
  symbolic component impose no additional constraint by
  Lemma~\ref{lemma:one-step-neg2}, and the fact that, by
  Lemma~\ref{lemma:one-step}, the relation $Q\leq S$ is correctly
  propagated to all derived concolic states.
  The fact that the traces are the same follows trivially by
  Lemma~\ref{lemma:one-step-neg2}.
\end{proof}
Before proving Theorem~\ref{theorem:sound}, we need the following
auxiliary result:

\begin{lemma} \label{lemma:sound}
  Let $\tuplec{Q\sep S}$ be a concolic state with
  $S=\mystate{d}{\ol{A}}$. If there is a concolic execution
  of the form
  $\tuplec{Q\sep S} \stackrel{r}{\lracn}_{\pi,R_Q,R_S} \tuplec{Q'\sep
    \mystate{d'\wedge d}{\ol{B}}}$, then $S\wedge d'\lra_r S''$ with
  $S''\equiv \mystate{d'\wedge d}{\ol{B}}$.
  Furthermore, $\rules(\catom(S\wedge d')) = R_Q$.
\end{lemma}

\begin{proof}
  Let $Q = \mystate{c}{p(\ol{u}),\ol{A}}$ and
  $S=\mystate{c'}{p(\ol{u'},\ol{A'}}$. Let
  $r' = (p(\ol{s})\leftarrow d\land \ol{B})$ be a fresh variant of
  rule $r$. Then, we have
  $Q' = \mystate{\ol{s}=\ol{u}\land d\land c}{\ol{B},\ol{A}}$ and
  $S' = \mystate{\ol{s}=\ol{u'}\land d\land c'}{\ol{B},\ol{A'}}$ (we
  considered the same renaming $r'$ of $r$ for simplicity).
  Let us now consider
  $(S\land d') = \mystate{\ol{s}=\ol{u'}\land d\land
    c'}{p(\ol{u'},\ol{A'}}$.  Let
  $r'' = (p(\ol{s'})\leftarrow d'\land \ol{B'})$ be another fresh
  variant of rule $r$. 
  
  Then, we have
  $S'' = \mystate{\ol{s'}=\ol{u'}\land d'\land \ol{s}=\ol{u'}\land
    d\land c'}{\ol{B'},\ol{A'}}$. Trivially, we have $S'\equiv S''$.
  The fact that $\rules(\catom(S\wedge d')) = R_Q$ then follows
  trivially by Lemma~\ref{lemma:one-step-neg2}.
\end{proof}
Now, we can  prove the soundness of concolic execution:

%\begin{theorem}[soundness] \label{theorem:sound}
  \mbox{}\\
  \textit{Theorem~\ref{theorem:sound} (soundness)}\\
  Let $\tuplec{Q\sep S_\epsilon}$ be an initial concolic state with
  $\tuplec{Q\sep S_\epsilon} \lracn^\ast \tuplec{Q' \sep S'_\pi}$. Let
  $S = \mystate{\true}{\ol{A}}$ and $S' = \mystate{d}{\ol{B}}$. Then,
  we have $\mystate{d}{\ol{A}} \lra^\ast S''$ such that $S'\equiv S''$
  and the associated trace is $\pi$.
%\end{theorem}

\begin{proof} %[Proof of Theorem~\ref{theorem:sound}]
  The proof is a simple induction on the length of the concolic
  execution derivation using Lemma~\ref{lemma:sound}.
\end{proof}

%%%%%%%%%%%%%%%%%%%%%%%%%%%%%%%%%%%%%%%%%%%%%%%%%
\section{Some Examples of Concolic Testing} \label{apendix:examples}

In this section, we show some detailed examples of the use of function
$\alts$ from Section~\ref{sec:concolic-testing-proc} as well as an
example of the concolic testing procedure.

\begin{example}[$CLP(\cT{erm})$] % \label{ex-alt-term}
  Let $I = \mystate{true}{p(W)}$ and $C=\mystate{c}{q(N)}$ with
  $c=(W=N)$. Let
  $H_1 = \mystate{X=a}{q(X)}$ and $H_2 = \mystate{\true}{q(s(M))}$. For
  brevity, we remove the occurrences of $\true$ in the formul\ae{}
  below.
  \begin{itemize}
  \item Let $\cH^+ = \emptyset$ and $\cH^- = \{H_1, H_2\}$.  Then, we have
    $\gamma_1 = \negcon(C,\cH^-) = \forall X\; (N\neq X\vee X\neq
    a)\land \forall M\; (N\neq s(M))$. As $\cD \models_v (c\land\gamma_1)$
    holds for any valuation $v$ with $\{W\mapsto b,N\mapsto b\}\subseteq v$,
    $c\land\gamma_1$ is satisfiable.

  \item Let $\cH^+ = \{H_1\}$ and $\cH^- = \{H_2\}$.  Then, we have
    $\gamma_2 = \negcon(C,\cH^-) = \forall M (N\neq s(M))$. As
    $\cD \models_v (c\land\gamma_2)$ holds for any valuation
    $v$ with $\{W\mapsto a,N\mapsto a\}\subseteq v$,
    $c\land\gamma_2$ is satisfiable.

  \item Let $\cH^+ = \{H_2\}$ and $\cH^- = \{H_1\}$.  Then, we have
    $\gamma_3 = \negcon(C,\cH^-) = \forall X (N\neq X\vee X\neq a)$. As
    $\cD \models_v (c\land\gamma_3)$ holds for any valuation
    $v$ with $\{W\mapsto b,N\mapsto b\} \subseteq v$,
    $c\land\gamma_3$ is satisfiable.

  \item Finally, let $\cH^+ = \{H_1, H_2\}$ and $\cH^- = \emptyset$.
    Then, $\gamma_4 = \negcon(C,\cH^-) = \true$.
  \end{itemize}
  Now, let us consider the following call:
  $\alts(I,C,\{H_2\},\{H_1,H_2\})$. According to the definition of
  function $\alts$, we should consider the following three
  possibilities:
  \begin{itemize}
  \item $\cH^+ = \emptyset$ and $\cH^- = \{H_1,H_2\}$. Since $\cH^+$
    is empty, we can immediately conclude that
    $I \land c \land \gamma_1 \in \alts(I,C,\{H_2\},\{H_1,H_2\})$,
    \ie we produce the following state:
    \[
      \mystate{W=N\land \forall X\; (N\neq X\vee X\neq a)\land
        \forall M\; (N\neq s(M))}{p(W)}
    \]
    which could be simplified to
    $\mystate{\forall M\; (W\neq a\land W\neq s(M))}{p(W)}$.

  \item $\cH^+ = \{H_1\}$ and $\cH^- = \{H_2\}$.  Here, we should
    check that $C\land\gamma_2\approx H_1$ holds, which is true since
    $C\land\gamma_2 = \mystate{W=N\land \forall M\; (N\neq s(M))}{q(N)}$ and
    $\mathit{solv}(N=X\land X=a\land W=N\land \forall M\; (N\neq
    s(M)))=\true$. Therefore, we have
    $I\land c\land \gamma_2 \in \alts(I,C,\{H_2\},\{H_1,H_2\})$, \ie
    we produce the following state:
    \[
      \mystate{W=N\land \forall M\; (N\neq s(M))}{p(W)}
    \]
    which could be simplified to
    $\mystate{\forall M\; W\neq s(M)}{p(W)}$.
  \item $\cH^+ = \{H_1,H_2\}$ and $\cH^- = \emptyset$. In this case,
    we should check that $C\approx H_1$ and $C\approx H_2$, which is
    true. Therefore, $I\land c \in \alts(I,C,\{H_2\},\{H_1,H_2\})$,
    where $I\land c = \mystate{W=N}{p(W)}$, which could be
    simplified to $\mystate{\true}{p(W)}$.
  \end{itemize}
  To summarize, in this case we have
  \[
    \begin{array}{l@{~}l@{~}l}
    \alts(I,C,\{H_2\},\{H_1,H_2\}) = \{ & \mystate{\forall M\; (W\neq
                                          a\land W\neq s(M))}{p(W)}, \\
      & \mystate{\forall M\; W\neq
        s(M)}{p(W)}, \mystate{\true}{p(W)}
        & \}
    \end{array}
    \]
\end{example}

\pagebreak

\begin{example}[CLP($\cN)$]
  % \label{ex-alt-Qlin}
  Let $I = \mystate{\true}{p(W)}$,
  $C = \mystate{c}{q(X)}$,
  $\cH_Q = \{H_1\}$ and $\cH_S = \{H_1,H_2,H_3\}$, with
  $c=(W=X\land X \le 10)$,
  $H_1 = \mystate{Y\le 2}{q(Y)}$,
  $H_2 = \mystate{8\le Z\le 10}{q(Z)}$ and
  $H_3 = \mystate{T < 5}{q(T)}$.  Consider the case
  $\cH^+ = \{H_1,H_2\}$ and $\cH^- = \{H_3\}$.

  First, we should compute $\gamma = \negcon(C,\cH^-)$, \ie
  $\forall T\; (X\neq T\vee 5\le T)$, which can be simplified to
  $\gamma = (5\le X)$. So,
  $c\land\gamma=(W=X\land X \le 10 \land 5\le X)$
  can be simplified to $c\land\gamma=(W=X\land 5\le X \le 10)$,
  which is clearly satisfiable. Now, we should check that
  $C\land \gamma = \mystate{W=X\land 5\le X \le 10}{q(X)}$
  unifies with both $H_1$ and $H_2$ in order to produce an element of
  $\alts(I,C,\cH_Q,\cH_S)$:
  \begin{itemize}
  \item $C\land\gamma \approx H_1$. In this case, we have
    $\mathit{solv}(X=Y\land  Y\le 2\land W=X\land 5\le X
    \le 10) = \false$.
  \item $C\land\gamma \approx H_2$. In this case, we have
    $\mathit{solv}(X=Z\land 8\le Z\le 10\land W=X\land 5\le X
    \le 10) = \true$ (consider, \eg any valuation $v$ with
    $\{X\mapsto 9,Z\mapsto 9, W\mapsto 9\}\subseteq v$).
  \end{itemize}
  Therefore, this case is not feasible.

  Let us now consider instead the case $\cH^+ = \{H_1,H_3\}$ and
  $\cH^- = \{H_2\}$.
  First, we should compute $\gamma' = \negcon(C,\cH^-)$, \ie
  $\forall Z\; (X\neq Z\vee Z<8\vee Z>10)$, which can be simplified to
  $\gamma' = (X<8\vee X>10)$. So,
  $c\land\gamma'=(W=X\land  X \le 10 \land (X<8\vee X>10))$
  can be simplified to $c\land\gamma'=(W=X\land  X < 8)$,
  which is clearly satisfiable. Now, we should
  check that $C\land \gamma'$ unifies with both $H_1$ and $H_3$ in
  order to produce an element of $\alts(I,C,\cH_Q,\cH_S)$:
  \begin{itemize}
  \item $C\land\gamma' \approx H_1$. In this case, we have
    $\mathit{solv}(X=Y\land  Y\le 2\land W=X\land  X < 8)
    =\true$ (consider, \eg any valuation $v$ with
    $\{X\mapsto 1,Y\mapsto 1,W\mapsto 1\}\subseteq v$).

 \item $C\land\gamma' \approx H_3$. In this case, we have
    $\mathit{solv}(X=T\land T<5\land W=X\land  X < 8) =\true$
    (consider, \eg any valuation $v$ with
    $\{X\mapsto 4,T\mapsto 4,W\mapsto 4\}\subseteq v$).
  \end{itemize}
  Therefore, we have
  $I\land c \land \gamma'\in\alts(I,C,\cH_S,\cH_Q)$, \ie
  we produce the state:
  \[
    \mystate{W=X\land  X < 8}{p(W)}   %%\;.
  \]
  which can be simplified to $\mystate{W < 8}{p(W)}$.
\end{example}

\begin{example}[concolic testing]
  Consider again the CLP($\cT{erm}$) program of Example~\ref{ex:concolicexec}.
  Given $\mystate{N=a}{p(N)}$ as the initial concrete state, concolic
  testing starts with the following initial configuration:
  \[
    (\emptyset,\{\mystate{N=a}{p(N)}\},\emptyset,\mystate{\true}{p(N')},
    \tuplec{\mystate{N=a}{p(N)}\sep \mystate{\true}{p(N')}_\epsilon})
  \]
  Let $\TC_0 = \{\mystate{N=a}{p(N)}\}$ and
  $I=\mystate{\true}{p(N')}$.
  Then, concolic testing proceeds as follows:
  {\[
    \begin{array}{l}
    (\emptyset,\TC_0,\emptyset,I,
      \tuplec{\mystate{N=a}{p(N)}\sep \mystate{\true}{p(N')}_\epsilon})\\
      \leadsto_\blue{\sf alts (choice)}
    (\PTC_1,\TC_0,\{\epsilon\},I,
      \tuplec{\mystate{N=a}{p(N)}^{r_1}\sep
      \mystate{\forall Y' (N'\neq s(Y'))}{p(N')}_\epsilon^{r_1}})\\
      \leadsto_\blue{\sf skip (unfold)}
    (\PTC_1,\TC_0,\{\epsilon\},I,
      \tuplec{\mystate{X=N\land X=a\land N=a}{\epsilon}\\
      \hspace{12ex}\sep\mystate{X'=N'\land X'=a\land \forall Y' (N'\neq s(Y'))}{\epsilon}_{\ell_1}})\\
      \leadsto_\blue{\sf restart}
    (\PTC_2,\TC_1,\{\epsilon\},I,
      \tuplec{\mystate{\true}{p(N')}\sep
      \mystate{\true}{p(N')}_{\epsilon}})\\
      \leadsto_\blue{\sf skip (choice)}
    (\PTC_2,\TC_1,\{\epsilon\},I,
      \tuplec{\mystate{\true}{p(N')}^{r_1},
      \mystate{\true}{p(N')}^{r_2}\\
      \hspace{12ex}\sep
      \mystate{\true}{p(N')}_{\epsilon}^{r_1}, \mystate{\true}{p(N')}_{\epsilon}^{r_2}})\\
      \leadsto_\blue{\sf skip (unfold)}
    (\PTC_2,\TC_1,\{\epsilon\},I,
      \tuplec{\mystate{X=N'\land X=a}{\epsilon},
      \mystate{\true}{p(N')}^{r_2}
      \\ \hspace{12ex}\sep
      \mystate{X'=N'\land X'=a}{\epsilon}_{\ell_1}, \mystate{\true}{p(N')}_{\epsilon}^{r_2}})\\
      \leadsto_\blue{\sf skip (next)}
    (\PTC_2,\TC_1,\{\epsilon\},I,
      \tuplec{\mystate{\true}{p(N')}^{r_2}\sep
      \mystate{\true}{p(N')}_{\epsilon}^{r_2}})\\
      \leadsto_\blue{\sf skip (unfold)}
    (\PTC_2,\TC_1,\{\epsilon\},I,
      \tuplec{\mystate{s(Y)=N'}{q(Y)}\sep
      \mystate{s(Y')=N'}{q(Y')}_{\ell_2}})\\
      \leadsto_\blue{\sf alts (choice)}
    (\PTC_3,\TC_1,\{\epsilon,\ell_2\},I,
      \tuplec{\mystate{s(Y)=N'}{q(Y)}^{r_3}\sep
      \mystate{s(Y')=N'}{q(Y')}_{\ell_2}^{r_3}})\\
      \leadsto_\blue{\sf skip (unfold)}
    (\PTC_3,\TC_1,\{\epsilon,\ell_2\},I,
      \tuplec{\mystate{W=Y\land W=a\land s(Y)=N'}{\epsilon}\\
      \hspace{12ex}\sep
      \mystate{W'=Y'\land W'=a\land s(Y')=N'}{\epsilon}_{\ell_2\ell_3}})\\
      \leadsto_\blue{\sf restart}
    (\PTC_4,\TC_2,\{\epsilon,\ell_2\},I,
      \tuplec{\mystate{N'\neq a}{p(N')}\sep
      \mystate{\true}{p(N')}_{\epsilon}})\\
      \leadsto_\blue{\sf skip (choice)}
    (\PTC_4,\TC_2,\{\epsilon,\ell_2\},I,
      \tuplec{\mystate{N'\neq a}{p(N')}^{r_2}\sep
      \mystate{\forall X' (X'\neq N'\vee X'\neq a)}{p(N')}_{\epsilon}^{r_2}})\\
      \leadsto_\blue{\sf skip (unfold)}
    (\PTC_4,\TC_2,\{\epsilon,\ell_2\},I,
      \tuplec{\mystate{s(Y)=N'\land N'\neq a}{q(Y)}\\
      \hspace{12ex}\sep
      \mystate{s(Y')=N'\land \forall X' (X'\neq N'\vee X'\neq a)}{q(Y')}_{\ell_2}})\\
      \leadsto_\blue{\sf skip (choice)}
    (\PTC_4,\TC_2,\{\epsilon,\ell_2\},I,
      \tuplec{\mystate{s(Y)=N'\land N'\neq a}{q(Y)}^{r_3}\\
      \hspace{12ex}\sep
      \mystate{s(Y')=N'\land \forall X' (X'\neq N'\vee X'\neq a)}{q(Y')}_{\ell_2}^{r_3}})\\
      \leadsto_\blue{\sf skip (unfold)}
    (\PTC_4,\TC_2,\{\epsilon,\ell_2\},I,
      \tuplec{\mystate{W=Y\land W=a\land s(Y)=N'\land N'\neq
      a}{\epsilon}\\
      \hspace{12ex}\sep
      \mystate{W'=Y'\land W'=a\land s(Y')=N'\land \forall X' (X'\neq N'\vee X'\neq a)}{\epsilon}_{\ell_2\ell_3}})\\
      \leadsto_\blue{\sf restart}
    (\PTC_5,\TC_3,\{\epsilon,\ell_2\},I,
      \tuplec{\mystate{N'\neq a \land \forall Y' (N'\neq s(Y'))}{p(N')}\sep
      \mystate{\true}{p(N')}_{\epsilon}})\\
      \leadsto_\blue{\sf restart}
    (\PTC_6,\TC_4,\{\epsilon,\ell_2\},I,
      \tuplec{\mystate{\forall W' (Y'\neq W'\vee W'\neq a)\land s(Y')=N'}{p(N')}\sep
      \mystate{\true}{p(N')}_{\epsilon}})\\
      \leadsto_\blue{\sf skip (choice)}
    (\PTC_6,\TC_4,\{\epsilon,\ell_2\},I,
      \tuplec{\mystate{\forall W' (Y'\neq W'\vee W'\neq a)\land
      s(Y')=N'}{p(N')}^{r_2}\\
      \hspace{12ex}\sep
      \mystate{\forall X' (X'\neq N'\vee  X'\neq a)}{p(N')}_{\epsilon}^{r_2}})\\
      \leadsto_\blue{\sf skip (unfold)}
    (\PTC_6,\TC_4,\{\epsilon,\ell_2\},I,
      \tuplec{\mystate{s(Y)=N'\land \forall W' (Y'\neq W'\vee W'\neq
      a)\land s(Y')=N'}{q(Y)}\\
      \hspace{12ex}\sep
      \mystate{s(Y')=N'\land\forall X' (X'\neq N'\vee  X'\neq a)}{q(Y')}_{\ell_2}})\\
      \not\leadsto
    \end{array}
  \]}
  where
  {
    \[
    \begin{array}{l}
      \PTC_1 = \{\mystate{\true}{p(N')},\mystate{N'\neq a}{p(N')},
      \mystate{N'\neq a \land \forall Y' (N'\neq s(Y'))}{p(N')}\}\\
      \PTC_2 = \{\mystate{N'\neq a}{p(N')},
      \mystate{N'\neq a \land \forall Y' (N'\neq s(Y'))}{p(N')}\}\\
      \TC_1 = \{ \mystate{\true}{p(N')}, \mystate{N=a}{p(N)}\}\\
      \PTC_3 = \{\mystate{N'\neq a}{p(N')},
      \mystate{N'\neq a \land \forall Y' (N'\neq s(Y'))}{p(N')},\\
      \hspace{9ex}\mystate{\forall W' (Y'\neq W'\vee W'\neq a)\land s(Y')=N'}{p(N')} \}\\
      \PTC_4 = \{
      \mystate{N'\neq a \land \forall Y' (N'\neq s(Y'))}{p(N')},
      \mystate{\forall W' (Y'\neq W'\vee W'\neq a)\land s(Y')=N'}{p(N')} \}\\
      \TC_2 = \{ \mystate{N'\neq a}{p(N')},\mystate{\true}{p(N')}, \mystate{N=a}{p(N)}\}\\
      \PTC_5 = \{
      \mystate{\forall W' (Y'\neq W'\vee W'\neq a)\land s(Y')=N'}{p(N')} \}\\
      \TC_3 = \{ \mystate{N'\neq a \land \forall Y' (N'\neq s(Y'))}{p(N')},\mystate{N'\neq a}{p(N')},\mystate{\true}{p(N')}, \mystate{N=a}{p(N)}\}\\
      \PTC_6 = \{\:\}\\
      \TC_4 = \{ \mystate{\forall W' (Y'\neq W'\vee W'\neq a)\land s(Y')=N'}{p(N')},
      \mystate{N'\neq a \land \forall Y' (N'\neq s(Y'))}{p(N')},\\
      \hspace{9ex}\mystate{N'\neq a}{p(N')},\mystate{\true}{p(N')}, \mystate{N=a}{p(N)}\}\\
    \end{array}
  \]}
  Therefore, the set of test cases produced by our algorithm is
  $\TC_4$, which cover all execution paths:
  \begin{itemize}
  \item test case $\mystate{N=a}{p(N)}$ follows the trace $\ell_1$;
  \item test case $\mystate{\true}{p(N')}$ follows the trace $\ell_1$,
    then backtracks, and finally follows trace $\ell_2\ell_3$;
  \item test case $\mystate{N'\neq a}{p(N')}$ follows the trace
    $\ell_2\ell_3$;
  \item test case
    $\mystate{N'\neq a \land \forall Y' (N'\neq s(Y'))}{p(N')}$
    matches no rule;
  \item finally, test case $ \mystate{\forall W' (Y'\neq W'\vee W'\neq
      a)\land s(Y')=N'}{p(N')}$ follows a trace $\ell_2$ and,  then, fails.
  \end{itemize}
\end{example}

%%%%%%%%%%%%%%%%%%%%%%%%%%%%%%%%%%%%%%%%%%%%%%%%%
% Connections with the CSUP
%%%%%%%%%%%%%%%%%%%%%%%%%%%%%%%%%%%%%%%%%%%%%%%%%
\section{Proofs for Section~\ref{sec:connections-csup}}

In this section, we show the proofs of some technical results from
Section~\ref{sec:connections-csup}.

% \begin{proposition}\label{prop:alts-in-p}
\mbox{}\\
\textit{Proposition~\ref{prop:alts-in-p}-1}\\
  Suppose that $c\land\gamma$ is satisfiable and $C\land\gamma$ unifies
  with each element of $\cH^+$. Then, we have $C\land\gamma \equiv C'$
  for some $C'\in\cP$.

%\end{proposition}
Note that $C\land\gamma\not\in\cP$. 
Indeed, as $\gamma$ contains
the variables of $\cH^-$, we have that $C\land\gamma$ is not variable
disjoint with $\cH^+ \cup \cH^-$, so the condition
``$C\land d$ is variable disjoint with $\cH^+ \cup \cH^-$"
in Def.~\ref{def:constraint_sup} does not hold for $C\land\gamma$.

\begin{proof} %[Proof of Proposition~\ref{prop:alts-in-p}]
  Let $\gamma'$ be a variant of $\gamma$ where the variables occurring
  in $\cH^-$ have been renamed to new, fresh, variables. Then
  $C\land \gamma'$ is variable disjoint with $\cH^-$. Moreover, as
  all the variables of $\cH^-$ are bound in $\gamma$,
  $c\land\gamma'$ is satisfiable and $C\land\gamma'$ unifies
  with each element of $\cH^+$. Also, Prop.~\ref{prop:gamma-do-not-unify}
  is valid for $C\land \gamma'$ \ie $C\land \gamma'$ does not
  unify with any constraint atom in $\cH^-$.
  Therefore, by Def.~\ref{def:constraint_sup} with $d=\gamma'$, we have
  $C\land\gamma' \in \cP$.
  Note that we also have $\set(C\land \gamma) = \set(C\land\gamma')$.
  Hence the result, with $C'=(C\land\gamma')$.
\end{proof}

% \begin{proposition}\label{prop:p-in-alts}
\mbox{}\\
\textit{Proposition~\ref{prop:alts-in-p}-2}\\
  For each $C'\in\cP$ we have $C' \leq (C\land\gamma)$.
%\end{proposition}

﻿So intuitively, $C\land\gamma$ is maximal. 

\begin{proof} %[Proof of Proposition~\ref{prop:p-in-alts}]
  By Def.~\ref{def:constraint_sup} and Prop.~\ref{prop:gamma-maximality}.
\end{proof}

% \begin{proposition}\label{prop:if-P-neq-emptyset}
\mbox{}\\
\textit{Proposition~\ref{prop:alts-in-p}-3}\\
  If $\cP\neq\emptyset$ then $C\land\gamma$ unifies
  with each constraint atom in $\cH^+$.
%\end{proposition}

\begin{proof} %[Proof of Proposition~\ref{prop:if-P-neq-emptyset}]
  Suppose that $\cP\neq\emptyset$. Let $C\land d \in\cP$.
  Then, by Def.~\ref{def:constraint_sup} and
  Prop.~\ref{prop:gamma-maximality}, we have
  $\set(C\land d) \subseteq \set(C\land \gamma)$.
  Moreover, for each $H\in\cH^+$, as $C\land d$
  unifies with $H$ we have
  $\set(C\land d)\cap\set(H)\neq\emptyset$
  (by Lemma~\ref{lemma:unify-set}).
  Hence, for each $H\in\cH^+$, we have
  $\set(C\land \gamma)\cap\set(H)\neq\emptyset$ \ie
  $C\land \gamma$ unifies with $H$.
\end{proof}

% \begin{theorem}\label{thm:if-C-and-gamma-unifies}
\mbox{}\\
\textit{Theorem~\ref{thm:if-C-and-gamma-unifies}}\\
  If $C\land\gamma$ unifies with each element of $\cH^+$
  then $\set(C\land\gamma) = \set(\cP)$.
%\end{theorem}

\begin{proof} %[Proof of Theorem~\ref{thm:if-C-and-gamma-unifies}]
  Suppose that $C\land\gamma$ unifies with each element of $\cH^+$.
  \begin{itemize}
    \item If $c\land\gamma$ is not satisfiable, then we have
    $\set(C\land\gamma)=\emptyset$. Consequently,
    by Prop.~\ref{prop:alts-in-p}-2 we have $\set(C')=\emptyset$
    for all $C'\in\cP$. Therefore, we have
    $\set(\cP)=\emptyset=\set(C\land\gamma)$.
    \item If $c\land\gamma$ is satisfiable, then, by
    Prop.~\ref{prop:alts-in-p}-1 we have
    $\set(C\land\gamma) \subseteq \set(\cP)$.
    By Prop.~\ref{prop:alts-in-p}-2, we also have
    $\set(\cP) \subseteq \set(C\land\gamma)$.
    Hence the result.
  \end{itemize}
\end{proof}

% \begin{corollary}\label{cor:decidability}
\mbox{}\\
\textit{Corollary~\ref{cor:decidability}}\\
  The \emph{constraint selective unification problem} for $C$ with
  respect to $\cH^+$ and $\cH^-$ is decidable.
%\end{corollary}

\begin{proof} %[Proof of Corollary~\ref{cor:decidability}]
  We test whether $c \land \gamma$ is  satisfiable and
  $C\land\gamma$ unifies with each element of $\cH^+$.
  Both conditions are decidable because we assume that the constraint
  solver can decide any first-order formula of the constraint domain.\\
  If $C \land \gamma$ does not
  unify with one constraint atom of $\cH^+$ then $\cP = \emptyset$ by
  Prop.~\ref{prop:alts-in-p}-3.
  Otherwise, if $c \land \gamma$ is  not satisfiable,
  then $\set(C\land\gamma)=\emptyset$, and as $\set(C\land\gamma)=\set(\cP)$
  by Theorem~\ref{thm:if-C-and-gamma-unifies},
  we have $\set(\cP)=\emptyset$ hence $\cP=\emptyset$.
  Else, by Prop.~\ref{prop:alts-in-p}-1, we know that $\cP\neq\emptyset$.
  Note that in this latter case we know from Theorem~\ref{thm:if-C-and-gamma-unifies}
  that $\set(\cP)=\set(C\land\gamma)$.
\end{proof}

\end{document}